\documentclass[a4paper,USenglish]{lipics-v2021}
\usepackage{hyperref}

\pdfoutput=1 

\nolinenumbers
\title{Arbitrary Pattern Formation on a Continuous Circle by Oblivious Robot Swarm} 

\author{Brati Mondal}{Department of Mathematics, Jadavpur University, Kolkata-700032, India}{bratim.math.rs@jadavpuruniversity.in}{https://orcid.org/0009-0001-3017-9924}{(Optional) author-specific funding acknowledgements}

\author{Pritam Goswami}{Department of Mathematics, Jadavpur University, Kolkata-700032, India}{pritamgoswami.math.rs@jadavpuruniversity.in}{https://orcid.org/0000-0002-0546-3894}{[funding]}

\author{Avisek Sharma}{Department of Mathematics, Jadavpur University, Kolkata-700032, India}{aviseks.math.rs@jadavpuruniversity.in}{https://orcid.org/0000-0001-8940-392X}{[funding]}

\author{Buddhadeb Sau}{Department of Mathematics, Jadavpur University, Kolkata-700032, India}{buddhadeb.sau@jadavpuruniversity.in}{https://orcid.org/0000-0001-7008-6135}{[funding]}

\authorrunning{B.Mondal, P.Goswami, A.Sharma and B.Sau} 

\Copyright{Jane Open Access and Joan R. Public} 



\category{} 

\relatedversion{} 

\EventEditors{John Q. Open and Joan R. Access}
\EventNoEds{2}
\EventLongTitle{42nd Conference on Very Important Topics (CVIT 2016)}
\EventShortTitle{CVIT 2016}
\EventAcronym{CVIT}
\EventYear{2016}
\EventDate{December 24--27, 2016}
\EventLocation{Little Whinging, United Kingdom}
\EventLogo{}
\SeriesVolume{42}
\ArticleNo{23}
\usepackage{amsfonts, amssymb, amsmath,pifont,float,caption,multicol}
\usepackage[ruled,vlined,linesnumbered]{algorithm2e}
\begin{document}

\maketitle

\begin{abstract}
In the field of distributed system, Arbitrary Pattern Formation (APF) problem is an extensively studied problem. The purpose of APF is to design an algorithm to move a swarm of robots to a particular position on an environment (discrete or continuous) such that the swarm can form a specific but arbitrary pattern given previously to every robot as an input. 
In this paper the solvability of the APF problem on a continuous circle has been discussed for a swarm of oblivious and silent robots without chirality under a semi synchronous scheduler. Firstly a class of configurations called  \textit{Formable Configuration}($FC$) has been provided which is necessary to solve the APF problem on a continuous circle. Then considering the initial configuration to be an $FC$, an deterministic and distributed algorithm has been provided that solves the APF problem for $n$ robots on a continuous circle of fixed radius within $O(n)$ epochs without collision.
\end{abstract}
\section{Introduction}
\label{Sec:1}
Applications of distributed systems and their relevant problems have gained substantial importance in the last two decades. Unlike a centralized system, using a swarm of inexpensive, simple robots to do a task is more cost-effective, robust, and scalable. These swarms of robots have many applications, like rescue operations, military operations, search and surveillance, disaster management, cleaning large surfaces, and so on. 

Researchers are interested in studies about using swarm robots with minimum capabilities to do some specific tasks like \textsc{Gathering}, \textsc{Arbitrary Pattern Formation}, \textsc{Dispersion}, \textsc{Exploration}, \textsc{Scattering} etc. The robots are \textit{autonomous} (have no centralized controller), \textit{anonymous} (have no IDs), and \textit{homogeneous} (have the same capabilities and execute the same algorithm). Depending on the capabilities of robots, there are four types of robot models: $\mathcal{OBLOT}$, $\mathcal{FSTA}$, $\mathcal{FCOM}$, $\mathcal{LUMI}$. In the $\mathcal{OBLOT}$ model, robots do not have any persistent memory of their previous state, i.e. \textit{oblivious} and they can't communicate with each other, i.e. \textit{silent}. In the $\mathcal{FSTA}$ model, robots are not oblivious but silent. In the $\mathcal{FCOM}$ model, robots are oblivious but not silent. In the $\mathcal{LUMI}$ model, robots are neither oblivious nor silent.

Each robot executes a \textsc{Look-Compute-Move} (LCM) cycle after activation. In \textsc{Look} phase, the robot takes a snapshot of its surroundings and collects the required information. Then the robot calculates the target using that information in the \textsc{Compute} phase and moves to the destination in \textsc{Move} phase. A scheduler is the controller of the activation of robots. There are three types of schedulers: \textit{Fully synchronous} (\textsc{FSync}) scheduler, \textit{Semi synchronous} (\textsc{SSync}) scheduler, \textit{Asynchronous} (\textsc{ASync}) scheduler. In \textsc{FSync} scheduler, time is divided into global rounds of the same duration, and each robot is activated in every round and executes the LCM cycle. In \textsc{SSync} scheduler, time is also divided into global rounds of the same duration as \textsc{FSync}, but all robots may not be activated at the beginning of each round. In \textsc{ASync} scheduler, robots are activated independently, and the LCM cycle is not synchronised here.

The problem considered here is the Arbitrary Pattern Formation (APF) problem, in which a swarm of robots is deployed in an environment (discrete or continuous domain). The APF problem aims to design an algorithm such that robots move to a particular position and form a specific but arbitrary pattern, which is already given to every robot as input. There is a vast literature on APF in both discrete and continuous domains (\cite{BAKS20,BDS2021,BOSEKAS21,BQTS2016,BQTS2018,SAGA2023,SGA2019,CFSN20,DPV2010,FPSW2008,KGGS22,SY1996,YS2010}). Most of the works of APF in the continuous domain are considered on the Euclidean plane. There are other sorts of environments that are included in the continuous domain, e.g., any closed curve embedded on the plane where robots can only move on that curve. In real life, such environments also exist and are hugely applicable in different scenarios such as roads, railway tracks, tunnels, waterways, etc. Another example of this kind of environment is a circle of fixed radius embedded in a plane. Studying this problem is interesting because the solution can be extended to all other closed curves. Thus, in this paper, we have considered the problem of Arbitrary Pattern Formation (APF) on a circle. 

\section{Related Works and Our Contribution}
\label{Sec:2}
\subsection{Related work}
In swarm robotics, Arbitrary Pattern Formation (APF) problem is a hugely studied problem. This problem was first introduced by Suzuki and Yamashita in \cite{SY1996} on the Euclidean plane. Later, they characterised the geometric patterns formable by oblivious and anonymous robots in \cite{YS2010} for fully synchronous and asynchronous schedulers. After that, this problem has been considered in different environments on continuous and discrete domains (\cite{BAKS20,BDS2021,BOSEKAS21,BQTS2016,BQTS2018,SAGA2023,SGA2019,CFSN20,DPV2010,FPSW2008,KGGS22}). 

In the continuous domain, most of the works that consider arbitrary pattern formation are done on the Euclidean plane under different settings. In \cite{FPSW2008} Flocchini discussed the solvability of the pattern formation problem by considering oblivious robots with fully asynchronous schedulers. They showed that if the robots have no common agreement with their environment, they are unable to form an arbitrary pattern. Moreover, if the robots have one axis agreement, then any odd number of robots can form an arbitrary pattern, whereas the even number can't. Further, if the robots have both axis agreement, then any set of robots can form any pattern. They proved that it is possible to elect a leader if it is possible to form any pattern for $n \ge 3$ robots. The converse of this result is proved, and a relationship between leader election and arbitrary pattern formation for oblivious robots in the asynchronous scheduler is studied in \cite{DPV2010}. The authors showed that for $n \ge 4$ robots with chirality (respectively for $n \ge 5$ without chirality), the Arbitrary pattern formation problem is solvable if and only if leader election is possible. In \cite{BQTS2016} the authors proposed a probabilistic pattern formation algorithm for oblivious robots under an asynchronous scheduler without chirality. Their protocol is a combination of two phases: a probabilistic leader election phase and a deterministic pattern formation phase. Later in \cite{BQTS2018} they proposed a new geometric invariant that exists in any configuration with four oblivious anonymous mobile robots to solve arbitrary pattern formation problems with or without the common chirality assumption. In \cite{SGA2019} authors studied Embedded Pattern Formation without chirality with oblivious robots. They characterised when the problem can be solved by a deterministic algorithm and when it is unsolvable. In \cite{BDS2021} authors studied the APF problem for robots whose movements can be inaccurate and the formed pattern is very close to the given pattern. In \cite{BOSEKAS21} the authors provided a deterministic algorithm in the Euclidean plane with asynchronous opaque robots.


Note that all the work on arbitrary pattern formation considering the continuous domain has been done only for the Euclidean plane, where the robots can arbitrarily move from one point to another via infinitely many paths. But there are some environments in the continuous domain in which the movements of robots from one point to another are restricted to a finite number of possible paths. A continuous circle of fixed radius is one such environment. To the best of our knowledge, there are some works (\cite{CGKK11,CKPT14,FKKR17,FKKSY19,GSGS23,DUVY20}) which considered continuous circle as their corresponding environment. The problems of patrolling, gathering, and rendezvous are the main focus of these works. But none of them considered the problem of APF on the continuous circle.

\subsection{Our Contribution}
In this work, we aim to solve the problem of Arbitrary Pattern Formation (APF) on a continuous circle by oblivious and silent mobile robots with full visibility under a semi-synchronous scheduler. To the best of our knowledge, the APF problem has not yet been considered on a continuous circle. So in this paper, we have considered this problem for the first time. Here the robots do not agree with a particular direction i.e. robots have no chirality. The movements of robots are restricted only in two directions, clockwise and anti-clockwise from any point. So, avoiding collision in a circle is more difficult than avoiding collision on a plane.

The robot model considered here is the weakest $\mathcal{OBLOT}$ model. In this problem, there is no particular landmark or door from which the robots enter. Here we characterise the class of initial configuration for which this APF problem is solvable. We name this class of configurations, Formable Configuration (FC). We have shown that FC contains either asymmetric configuration or configuration having only reflectional symmetry and there exists at least one line of reflection having a robot on it. Then we have provided a deterministic and distributed algorithm $APF\_CIRCLE$ which solves this problem for any FC as the initial configuration within $O(n)$ epochs under a semi synchronous (\textsc{SSync}) scheduler.

Observe that, if a configuration is rotationally symmetric then it is not an FC and hence APF can't be solvable for this kind of configuration. So, maintaining a rotationally asymmetric configuration is necessary throughout an execution of any algorithm which solves the problem. There are some known techniques in \cite{BAKS20,DUVY20}, that maintain the asymmetry of the configuration. For example in an infinite grid, the asymmetry is maintained by moving a particular robot to a certain distance (\cite{BAKS20}). In a continuous circle we can not adapt this technique due to the bounded and circular nature of the environment. Furthermore, in \cite{DUVY20}, authors solved the problem of gathering under limited vision on a continuous circle. The algorithm proposed in this paper also requires to maintain  the rotational asymmetry of the configuration. For this purpose, the authors exploited the global weak multiplicity detection  and chirality agreement, in their work. But here in this work, we can not have the luxury of having multiplicity points as the target pattern does not include any multiplicity points and the robots are oblivious. So, for maintaining the asymmetry in circle, we have came up with a new technique that have been used in designing the deterministic algorithm $APF\_CIRCLE$, presented in this paper.

\subsection{Road map to the Paper}
Section~\ref{Sec:1} is dedicated to introducing the problem and then in Section~\ref{Sec:2} some related works and the contribution of this work has been described. In Section~\ref{Sec:3}, the problem definition along with the model is discussed. In Section~\ref{Sec:4}, Some preliminaries such as some definitions, notations, and some results have been established. In Section~\ref{Sec:5}, the leader election and target pattern embedding has been described. Section~\ref{Sec:6} is dedicated to describing the provided algorithm along with the correctness results and finally, this work is concluded in Section~\ref{Sec:7}.

\section{Model and Problem Definition }
\label{Sec:3}
\subsection{Problem Definition}
 Let $\mathcal{CIR}$ be a continuous circle of fixed radius. Let $n$, $n \ge 3$ robots resides on the perimeter of $\mathcal{CIR}$. The robots can move freely on the circle. A sequence of angular distances $\beta_0,\beta_1, \beta_2\dots, \beta_{n-1}$ is given to all the robots as input which is the target pattern such that the sum of the angles of the sequence is equal to $2\pi$. The problem is to design a distributed algorithm for the robots so that by finite execution of the algorithm, the robots move in such locations on $\mathcal{CIR}$, such that the final configuration has the following property:
 
 \begin{itemize}
     \item [\ding{71}] There exist a robot, say $r_0$ and a direction $\mathcal{D} \in \{$ clockwise, anticlockwise$\}$ such that the angular distance between the $i$-th and $(i+1)$-th robot in the direction $\mathcal{D}$ (denoted as $r_i$ and $r_{i+1}$ respectively) is $\beta_i$, where all the indices are considered under modulo $n$.
 \end{itemize}
 

\subsection{Model}

    \subsubsection{Robot Model:}  All robots are placed on the perimeter of a circle, say $\mathcal{CIR}$. Here the robots can move only on the perimeter of the circle. Robots have no particular orientation (i.e., no agreement on clockwise or anticlockwise direction). Robots have full visibility of the circle. The initial configuration is rotationally asymmetric. The movements of the robots are rigid i.e. robots always moves to its destination in a particular round. Robots have following properties-
\begin{itemize}
    \item[] \textit{Autonomous:} Robots don't have any centralised controller.
    \item[] \textit{Anonymous:} Robots have no IDs.
    \item[] \textit{Homogeneous:} All robots have same capabilities and execute the same algorithm.
    \item[] \textit{Oblivious:} Robots have no persistent memory.
    \item[] \textit{Silent:} Robots have no means of communications.
    \item[] \textit{Visibility:} Robots have full visibility of the circle i.e., robots can see all other robots on the circle.
\end{itemize}

\paragraph{\textbf{LCM cycle:}} Each robot executes a cycle of \textsc{Look-Compute-Move}(LCM) phases upon activation.

\begin{itemize}
    \item[] \textit{LOOK:} In look phase a robot takes a snapshot of its surroundings and gets the location of other robots on the circle according to its own coordinate .
    \item[] \textit{COMPUTE:} Robot determines target location using the snapshot of the look phase as the input by executing the provided algorithm.
    \item[] \textit{MOVE:}  In move phase robot moves to the destination point calculated in the compute phase.
\end{itemize}

\subsubsection{Scheduler Model:}The activation of robots are controlled by an entity called scheduler. Depending on the activation timing there are three types of schedulers: \\
\begin{itemize}
    \item[] \textit{Fully synchronous} (\textsc{FSync}): In fully synchronous scheduler time is divided into rounds of equal length and all robots are activated at the beginning of every round and performs the LCM cycle synchronously.\\
    
    \item[] \textit{Semi synchronous} (\textsc{SSync}): Similar to fully synchronous scheduler here also time is divided into rounds of equal length. But all robots might not get activated at the beginning of a particular round. In a particular round the activated robots perform the LCM cycle synchronously. Note that semi synchronous scheduler is more generalised than fully synchronous scheduler.\\
    
    \item[] \textit{Asynchronous} (\textsc{ASync}): In asynchronous scheduler time is not divided into rounds like fully synchronous and semi synchronous scheduler. Robots are activated independently. In a particular moment of time some robots may be in Look phase, some in Compute phase, some in Move phase or some may be idle. Asynchronous scheduler is the most general among all the schedulers. 
    
    In this paper, Semi Synchronous (\textsc{SSync}) scheduler has been considered to solve the problem.
\end{itemize}

\section{Prelimineries}
\label{Sec:4}
In this section, we first justify the reason for assuming the initial configuration being rotationally asymmetric. Before that let us first define  what a configuration is and what does it mean by the phrase ``configuration is rotationally asymmetric''.

\begin{definition}[$(A, B)_{\mathcal{D}}$]
    \label{def:angle}
    Let $A$ and $B$ be two points on the circle $\mathcal{CIR}$. Let $\mathcal{D}$ be a direction either clockwise or anticlockwise. Then $(A,B)_{\mathcal{D}}$  denotes the angular distance from point $A$ to point $B$ in the direction $\mathcal{D}$.
\end{definition}

\begin{definition}[Configuration]
    \label{def:config}
    A configuration is a set of $n$ points on the circle  $\mathcal{CIR}$ such that each point in the set consists exactly one robot.
\end{definition}

\begin{definition}[Rotationally Asymmetric Configuration]
    \label{def:asymmetry}
    A configuration is called a rotationally asymmetric configuration if there does not exist any non-trivial (less than $2\pi$) rotation of the circle $\mathcal{CIR}$ such that the configuration remains same after and before the rotation.
\end{definition}
\begin{definition}[Reflectionally Symmetric Configuration]
    \label{def:reflectionalSym}
    We call a configuration reflectionlly symmetric if
    \begin{enumerate}
        \item there exist a straight line called "\textit{line of reflection}" , say $L$, passing through the center of the circle and intersecting it at two points, say $A$ and $B$. $L$ divides the circle into two halves.
        \item for any robot $r$ located on any of the two halves there exist another robot $r'$ on the other half such that $(r,x)_{\mathcal{D}} = (r',x)_{\mathcal{D}'}$ where, $D$ is either the clockwise or, the anticlockwise direction, $D'$ is the opposite direction of $D$ and $x \in \{A,B\}$ 
    \end{enumerate}
\end{definition}

A configuration is called an \textit{Asymmetric} if the configuration is both rotationally and reflectionally asymmetric.
\begin{proposition}
    \label{Prop:rotationallySymmetricImpossible} There is no deterministic distributed algorithm that solves arbitrary pattern formation problem on a continuous circle if the initial configuration is rotationally symmetric.
\end{proposition}
\begin{proof}
Let $\mathcal{C}(0)$ be the initial configuration which is rotationally symmetric. Let $\mathcal{C}(0)$ has a $k-$fold symmetry i.e., for a rotation of $\frac{2\pi}{k}$ along the center, the configuration remains the same. By proposition 2.4 in \cite{DUVY20} if the scheduler is fully synchronous then for any algorithm $\mathcal{A}$, the new configuration will have a $k'-$fold symmetry after one execution of $\mathcal{A}$, where $k' \ge k$. Thus for any finite execution of $\mathcal{A}$, the configuration will always have a $k_1-$fold symmetry where $k_1 \ge k$. So if the target pattern is asymmetric it can not be formed by the robot swarm by finite execution of $\mathcal{A}$. Hence the result.
\end{proof}

\begin{proposition}
    \label{prop:reflectionalSymmetryImpossible} There is no deterministic algorithm which solves arbitrary pattern formation on a continuous circle without chirality if the initial configuration is reflectionally symmetric and there is no robot on a line of reflection.
\end{proposition}
\begin{proof}
    Let there be a deterministic algorithm $\mathcal{A}$ that solves the arbitrary pattern formation on a continuous circle without chirality. Let $C$ be a configuration which has a line of reflectional symmetry $L$ in a configuration such that $L$ does not contain any robot on it. Then we show that execution of algorithm $\mathcal{A}$ cannot destroy the reflection symmetry without coalition. This will imply $\mathcal{A}$ cannot form an asymmetric target pattern starting from $\mathcal{C}$. Since the line $L$ does not contain any robot then the total number of robots present is even. Each robot can be paired with its reflectional image with respect to $L$. We denote a pair as $[r,r']$. Suppose the adversary maintains a semi-synchronous scheduler where in a round only one such pair of robot is activated. Suppose a such pair $[r,r']$ is activated. If on activation $r$ decides to move a point $p$, then since $r$ and $r'$ have same view due to reflectional symmetry, so destination points of $r$ and $r'$ either are mirror images of each other with respect to $L$ or on the same point of $L$. If the destination points are on a point on $L$, then coalition takes place. Otherwise, after the completion of the move the new configuration still have the line of reflectional symmetry $L$. Therefore, the reflectional symmetry remains in every round if the mentioned activation schedule is considered.

\end{proof}

Next, we define some terminologies here in this section that will be needed to describe the algorithm provided in the next section.

\begin{definition}[ Set of Angle Sequences of a robot $r$] Let, $\mathcal{R}=$ $\{r_0, r_1, r_2,\dots,r _{n-1}\}$ be the set of robots placed consecutively in a fixed direction, say $\mathcal{D}$ (either clockwise or anticlockwise) on the circle $\mathcal{CIR}$. Let $\theta_{i \pmod{n}}$ be the angular distance from the location of robot $r_{i \pmod{n}}$ to the location of robot $r_{{i+1} \pmod{n}}$ in the direction $\mathcal{D}$. Then the set of angle sequences of the robot $r=r_0$, denoted as $\mathcal{AS}(r)$, is the set $\{\mathcal{AS}_D (r),\mathcal{AS}_{D'}(r)\}$ where, $\mathcal{D}'$ is the opposite direction of $\mathcal{D}$ and  $\mathcal{AS}_D (r)=(\theta_0, \theta_1,\dots,\theta_{n-1}$), $\mathcal{AS}_{D'} (r)=(\theta_{n-1}, \theta_{n-2},\dots,\theta_{0}$) are two angle sequences  in the direction $\mathcal{D}$ and $\mathcal{D'}$ respectively.
\end{definition}

Since the initial configuration is rotationally asymmetric then, for a fixed particular orientation (either clockwise or anti-clockwise) all the robots have different angle sequences \cite{DUVY20}. So, note that, for two robots, say $r_1$ and $r_2$ if $\mathcal{AS}_{\mathcal{D}_1}(r_1) \in \mathcal{AS}(r_1)$ is equal to $\mathcal{AS}_{\mathcal{D}_2}(r_2) \in \mathcal{AS}(r_2)$ then, $\mathcal{D}_1$ must be equals to $\mathcal{D}_2'$

\begin{definition}[Nominee] A robot $r_l$ is considered as the nominee if 
$$\min(\cup_{r\in \mathcal{R}} \mathcal{AS}(r)) \in \mathcal{AS}(r_l)$$   
\end{definition}

\begin{figure}
    \centering
    \includegraphics[height=4.5cm]{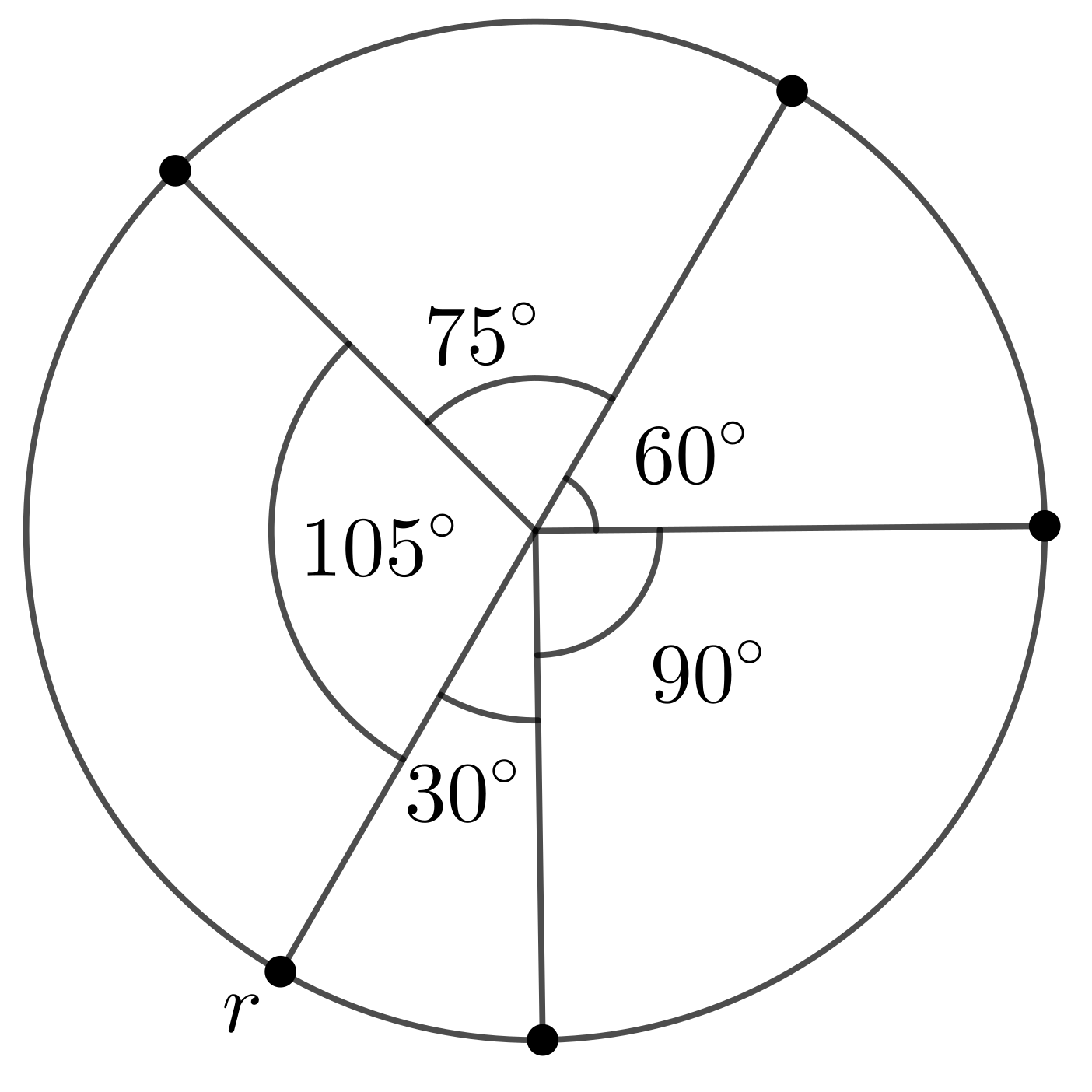}
    \caption{ Here, $\mathcal{AS}(r) = \{(30^{\circ} 90^{\circ} 60^{\circ} 75^{\circ} 105^{\circ}), (105^{\circ} 75^{\circ} 60^{\circ} 90^{\circ} 30^{\circ})\}$ is the set of angle sequences for the robot $r$. Note that $\mathcal{AS}(r)$ contains the minimum angle sequence $(30^{\circ} 90^{\circ} 60^{\circ} 75^{\circ} 105^{\circ})$ so, $r$ is a nominee. This configuration is also a single nominee configuration.} 
    \label{fig:SingleNominee}
\end{figure}

\begin{proposition}

\label{prop:atLeastoneAtmostTwo}
If the initial configuration is rotationally asymmetric then there are at least one and at most two nominees.
\end{proposition}

\begin{proof}
Since the initial configuration is rotationally asymmetric, by a result stated in \cite{DUVY20} it can be said that all the robots have distinct angle sequences in a particular direction. Thus $n$ robots have $n$ distinct angle sequences in a particular direction. Similarly in opposite direction, there exists $n$ distinct angle sequence. Among those $2n$ angle sequences, at least one angle sequence must be minimum. If this minimum angle sequence belongs to $\mathcal{AS}(r)$, then $r$ is selected as the nominee. So the initial configuration must have at least one nominee.
\par Now, let it be assumed that there are more than two nominees in the initial configuration. Without loss of generality let there be three nominees in the initial configuration. Let the first nominee whose angle sequence is minimum, has the minimum angle sequence in a particular direction say, $\mathcal{D}$. Then the second nominee must get its minimum angle sequence in the direction $\mathcal{D}'$, the opposite direction of $\mathcal{D}$ (as no two robots can have minimum angle sequence in the same direction). Now, the third nominee must have its minimum angle sequence in the direction of either $\mathcal{D}$ or $\mathcal{D}'$. But this can't be possible because, in a particular direction, no two robots have the same angle sequence. So there can not be more than two nominees in the initial configuration.

\end{proof}

\begin{figure}[h]
    \centering
    \includegraphics[height=4.5cm]{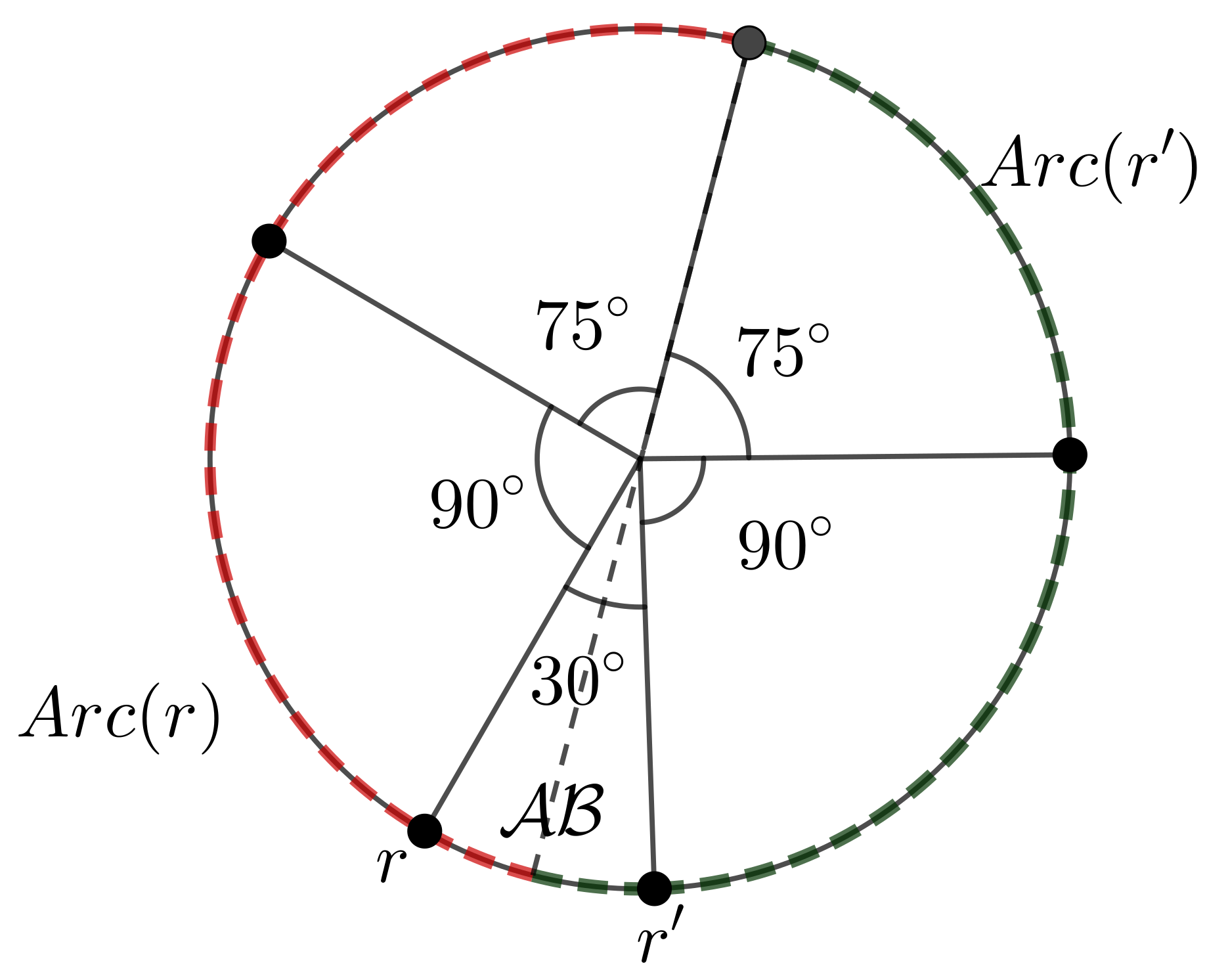}
    \caption{ A double nominee configuration where both $r$ and $r'$ are nominees. The angle bisector $\mathcal{AB}$ contains a robot. $Arc(r)$ is highlighted with red dotted line and $Arc(r')$ is highlighted with green dotted line.} 
    \label{fig:DoubleNominee}
\end{figure}

\begin{definition}[Single nominee configuration]
A rotationally asymmetric configuration is called a single nominee configuration if there is only one nominee.    
\end{definition}

\begin{definition}[Double nominee configuration]
A rotationally asymmetric configuration is called a double nominee configuration if there are two nominees.  
\end{definition}

\begin{definition}[Angle Bisector in a double nominee configuration] \label{def:angleBisector}
Let $r$ and $r'$ be two nominees in a double nominee configuration. The angle bisector of this configuration is defined as the straight line that bisects the angles formed by the robots $r$ and $r'$ and is denoted as $\mathcal{AB}$.
\end{definition}
In the future, the term ``angle bisector of a configuration" or the symbol $\mathcal{AB}$  will always be used for a double nominee configuration even if it is not mentioned explicitly. 

\begin{definition}[Arc of a nominee in a double nominee configuration]
    Let $r$ and $r'$ be two nominees in a double nominee configuration. Let $\mathcal{AB}$ be the angle bisector of the angle between $r$ and $r'$. Now, $\mathcal{AB}$ divides the circle into two arcs. Among these two arcs, the arc on which the robot $r$ is located except the points of $\mathcal{AB}$ is called the arc of the robot $r$ and is denoted as $Arc(r)$.
\end{definition}
\begin{proposition}
    \label{prop:asymISsinglenominee}
    An asymmetric configuration must be a single nominee configuration.
\end{proposition}
\begin{proof}
    Assume that the configuration is a double nominee configuration where $r_0$ and $r_0'$ be the two nominees.
    Let $L$ be the angle bisector $\mathcal{AB}$ which intersects the circle in two points, say $A$ and $B$. Let $\mathcal{D}$ be the direction such that $\mathcal{AS}_{\mathcal{D}}(r_0)$ is minimum in the configuration. Then $\mathcal{AS}_{\mathcal{D}}(r_0) =  \mathcal{AS}_{\mathcal{D}'}(r_0')$. Without loss of generality let $A$ be the first (among $A$ and $B$) point in the direction $\mathcal{D}$ from $r_0$. So, $A$ is also the first point among $A$ and $B$ from $r_0'$ in the direction $\mathcal{D}'$.  Thus, $(r_0,A)_{\mathcal{D}}=(r_0',A)_{\mathcal{D}'}$. Now, let $r$  be a robot on $Arc(r_0)$. Now, we have two cases.

    \textbf{Case I:} Let $r$ is on the arc joining $r_0$ and $A$ in the direction $\mathcal{D}$. Since $\mathcal{AS}_{\mathcal{D}}(r_0) =  \mathcal{AS}_{\mathcal{D}'}(r_0')$, there exists a robot $r'$ on $Arc(r_0')$ such that it is located on the arc joining $r_0'$ and $A$ in the direction $\mathcal{D}'$. Also, $(r_0,r)_{\mathcal{D}} = (r_0', r')_{\mathcal{D}'}$. So, $(r,A)_{\mathcal{D}} = (r_0,A)_{\mathcal{D}}-(r_0,r)_{\mathcal{D}} = (r_0',A)_{\mathcal{D}'}-(r_0',r')_{\mathcal{D}'} = (r',A)_{\mathcal{D}'}$. This is true for any $\mathcal{D} \in \{$clockwise direction, anticlockwise direction$\}$ and also for the point $B$. Hence the configuration has reflectional symmetry.

    \textbf{Case II:} Let $r$ is on the arc joining $r_0$ and $B$ in the direction $\mathcal{D}'$ from $r_0$. Since $\mathcal{AS}_{\mathcal{D}}(r_0) =  \mathcal{AS}_{\mathcal{D}'}(r_0')$, there exists a robot $r'$ on $Arc(r_0')$ such that it is located on the arc joining $r_0'$ and $B$ in the direction $\mathcal{D}$ and $(r_0,r)_{\mathcal{D'}} = (r_0', r')_{\mathcal{D}}$. Now, since $(r_0,A)_{\mathcal{D}} = (r_0',A)_{\mathcal{D}'}$, we have, $(r,A)_{\mathcal{D}} = (r,r_0)_{\mathcal{D}}+(r_0,A)_{\mathcal{D}} = (r',r_0')_{\mathcal{D}'}+(r_0',A)_{\mathcal{D}'} = (r',A)_{\mathcal{D}'}$. This is also true for any $\mathcal{D} \in \{$clockwise direction, anticlockwise direction$\}$ and also for the point $B$. Hence the configuration has reflectional symmetry.

    So, for both the cases we arrive at a contadiction. Thus an asymmetric configuratioin can not be a double nominee configuration.
\end{proof}

\begin{proposition}
\label{proposition:single nominee Is asymmetric}
    A single nominee configuration can not have reflectional symmetry.
\end{proposition}
\begin{proof}
    Let $\mathcal{C}$ be a single nominee configuration where $r_0$ be the nominee. If possible, let $\mathcal{C}$ has reflectional symmetry. Let $L$ be a line of reflection. We first claim that $r_0$ can not be on $L$. In this case let two neighbours of $r_0$ be $r_1$ in direction, say $\mathcal{D}$ and $r_{n-1}$ in $\mathcal{D'}$. Since $r_0$ is on $L$, $(R_0, R_1)_{\mathcal{D}} = (R_0, R_{n-1})_{\mathcal{D'}} =$ say $ \alpha_0$, where $R_0,R_1$ and $R_{n-1}$ are locations of $r_0, r_1$ and $r_{n-1}$ respectively. Also since $r_0$ is the nominee, $\alpha_0$ is the minimum angle in $\mathcal{C}$. Which implies $\mathcal{AS}_{\mathcal{D'}}(r_1)$ and $\mathcal{AS}_{\mathcal{D}}(r_{n-1})$ is smaller than $\mathcal{AS}_{\mathcal{D}}(r_0) = \mathcal{AS}_{\mathcal{D'}}(r_0)$ which is a contradiction. So, let $r_0$ is not on $L$. Let $r_0'$ be the reflection of $r_0$ along $L$. Now, if $r_0$ has its minimum angle sequence in the direction $\mathcal{D}$ then $\mathcal{AS}_{\mathcal{D}}(r_0) = \mathcal{AS}_{\mathcal{D'}}(r_0')$ which implies $\mathcal{C}$ is not an single nominee configuration contradicting our assumption. Thus we can conclude that a single nominee configuration must not have reflectional symmetry. 
\end{proof}

Now from definition, a single nominee configuration must be rotationally asymmetric. Thus from Proposition~\ref{proposition:single nominee Is asymmetric}, it can be concluded that a single nominee configuration must be asymmetric. Also, from Proposition~\ref{prop:asymISsinglenominee} we have the following theorem.  
\begin{theorem}
   \label{Theorem: Asym IFF single nominee} A configuration is asymmetric if and only if the configuration is a single nominee configuration. 
\end{theorem}

\begin{proposition}
    \label{prop: reflSymISdoublenominee}
    A configuration which is rotationally asymmetric, reflectionally symmetric must be a double nominee configuration .
\end{proposition}
\begin{proof}

    From Proposition~\ref{proposition:single nominee Is asymmetric} we have that if a configuration has reflectional symmetry then it can not be a single nominee configuration. Now, a rotationally asymmetric configuration can either be a single nominee configuration or a double nominee configuration. So, A rotationally asymmetric configuration having reflectional symmetry must be a double nominee configuration.
\end{proof}

Now from the above propositions (Proposition~\ref{prop:asymISsinglenominee} and Proposition~\ref{prop: reflSymISdoublenominee}) we can have the following result.
\begin{theorem}
    \label{thm: doublenominee iff rot Asy and reflect Sym}
    A configuration is a double nominee configuration if and only if the configuration is rotationally asymmetric and reflectionally symmetric.  
\end{theorem}

\begin{proposition}
    \label{prop:notLeaderConfigMeansRobotOnBisector}
   In a double nominee configuration with odd number of robots, exactly one robot must be located on the angle bisector $\mathcal{AB}$.
\end{proposition}
\begin{proof}
    By Theorem~\ref{thm: doublenominee iff rot Asy and reflect Sym}, a double nominee configuration must be roationally asymmetric, has reflectional symmetry and the angle bisector $\mathcal{AB}$ is the line of reflection without any nominee. Now, since the configuration has reflectional symmetry with respect to $\mathcal{AB}$, the total number of robots which are not on $\mathcal{AB}$ is even. Now, since it is given that the number of robots are odd, $\mathcal{AB}$ must contain exactly one robot.
    
\end{proof}

   

\section{Leader Election and Target Embedding}
\label{Sec:5}
\subsection{Leader Election} In this section, the concept of the leader in different configuration is discussed briefly. For embedding the target pattern on the circle $\mathcal{CIR}$, a particular point and a fixed direction must be agreed upon by all the robots of $\mathcal{R}$. The position of the elected leader is here used as the point.




For a single nominee configuration the unique nominee is considered to be the leader of the configuration. So, due to Proposition~\ref{prop:asymISsinglenominee}, for any asymmetric initial configuration, the unique nominee robot becomes the leader.

Now, if the configuration is symmetric then, there can be two types of symmetry.
\begin{itemize}
    \item rotational symmetry
    \item reflectional symmetry
\end{itemize}
 If a configuration has rotational symmetry then by Proposition~\ref{Prop:rotationallySymmetricImpossible} arbitrary pattern  formation is impossible. So, we assume that initial configuration is rotationally asymmetric. Thus by symmetric initial configuration only reflectional symmetry is considered. Now for a symmetric initial configuration let us consider the set $S_{Reflect}=\{r \in \mathcal{R} $ | $ r$ is on a line of reflection $\}$. Now for any $r \in S_{Reflect}$ we have an observation:

 \begin{observation}
 \label{obs:sameASbothDIRECTION}
 For any robot $r \in S_{Reflect},$ $\mathcal{AS}_{\mathcal{D}}(r) = \mathcal{AS}_{\mathcal{D}'}(r)$.
 \end{observation}
So, if $S_{Reflect} \ne \phi$ then, there must exist an unique $r_0 \in S_{Reflect}$ such that $\mathcal{AS}_{\mathcal{D}}(r_0) = \underset{r \in S_{Reflect}}{\min} \{\mathcal{AS}_{\mathcal{D}}(r) \cup \mathcal{AS}_{\mathcal{D}'}(r)\} = \underset{r \in S_{Reflect}}{\min} \mathcal{AS}_{\mathcal{D}}(r)$ (as the configuration is rotationally asymmetric $r_0$ is unique). So, for a symmetric configuration which is rotationally asymmetric but has reflectional symmetry, if $S_{Reflect} \ne \phi$ then we can elect a unique leader $r_0$, where $r_0 \in S_{Reflect}$ has the minimum angle sequence in $S_{Reflect}$. Note that, a double nominee configuration is rotationally asymmetric, has reflectional symmetry (due to Theorem~\ref{thm: doublenominee iff rot Asy and reflect Sym}). So, for a double nominee configuration, if $S_{Reflect} \ne \phi$ an unique leader can be elected from the set $S_{Reflect}$ as described earlier.
 
\begin{center}
\begin{tabular}{ |c|c|c|c| } 
\hline
\textbf{Configuration Classification} & \textbf{Leader}\\
\hline
Single nominee configuration & The unique nominee\\
\hline
Double nominee configuration & $r_0 \in S_{Reflect}$ such that\\ 
with $S_{Reflect} \ne \phi$ & $\mathcal{AS}_{\mathcal{D}}(r_0)=\underset{r \in S_{Reflect}}{\min} \mathcal{AS}_{\mathcal{D}}(r)$ \\
\hline
\end{tabular}
\end{center}
\subsection{Target Embedding}
From the above description it is clear that a unique leader can be elected for the mentioned configurations which is used to embed the  target pattern on $\mathcal{CIR}$  . Now the question is in which direction target pattern will be embedded. To answer this following two cases arise.

\textit{Case I:} Let $r_0$ be the elected leader in a single nominee configuration. Thus, $\mathcal{AS}_{\mathcal{D}}(r_0) \ne  \mathcal{AS}_{\mathcal{D}'}(r_0)$.
In this case, the angle sequences of the leader $r_0$ in the directions $\mathcal{D}$ and $\mathcal{D}'$ are different. So there must exists one direction in which the angle sequence of $r_0$ is minimum. This direction will be considered by all the robots for embedding the target pattern. Let the position of $r_0$ be denoted as $T_0$ and the direction in which $r_0$ has the smallest angle sequence is denoted as $\mathcal{D}$. Let the $j$-th target location on $\mathcal{CIR}$ from $T_0$ in the direction $\mathcal{D}$ is denoted by $T_j$, where $j \in \{0, 1, \dots, n-1\}$. The points are embedded in such a way on $\mathcal{CIR}$ that $(T_j, T_{j+1})_{\mathcal{D}} = \beta_j$ where the sequence $\beta_0, \beta_1, \dots, \beta_{n-1}$ is lexicographically smallest upto rotation of the input pattern given to the robots (all the indices are considered in modulo $n$).

\begin{figure}
    \centering
    \includegraphics[height=7cm]{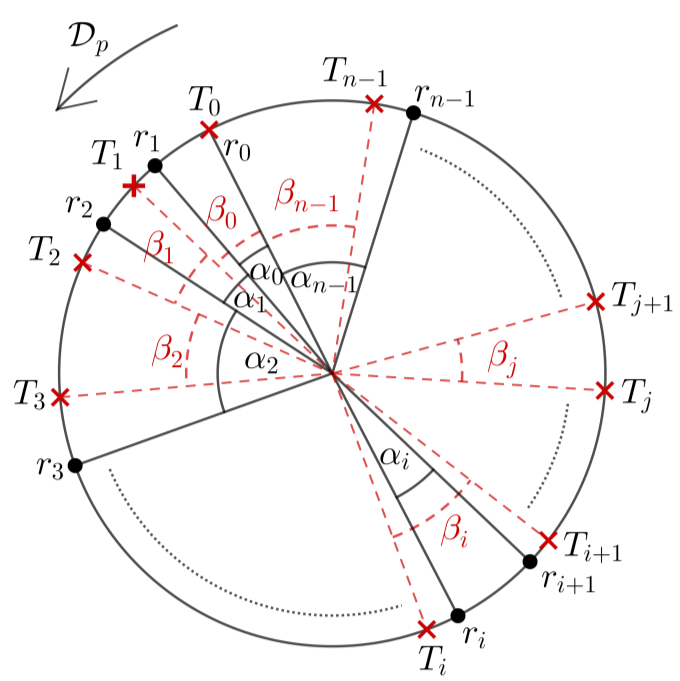}
    \caption{ $(\beta_0, \beta_1,\dots \beta_{n-1})$ is the smallest in lexicographic ordering of all possible sequences that can be formed from the input upto rotation. The sequence is embedded on the circle starting from the location of $r_0$ and in the direction $\mathcal{D}_p$. }
    \label{fig:targetembedding}
\end{figure}

\textit{Case II:} Let $r_0$ be the elected leader in a double nominee configuration. Then $\mathcal{AS}_{\mathcal{D}}(r_0) =  \mathcal{AS}_{\mathcal{D}'}(r_0)$.
In this case the leader $r_0$ has same angle sequence in both the direction $\mathcal{D}$ and $\mathcal{D}'$. So a particular direction can't be agreed upon by the other robots. Thus the target is embedded as described similar to the \textit{Case I}, but in both clockwise and anti-clockwise direction. Observe that there are two possible embedding for this case. 

We say that, target pattern is formed, if $n$ robots are on the $n$ points $T_j$, where $j \in \{0, 1, \dots, n-1\}$, for at least one such embedding.

\section{Algorithm $APF\_CIRCLE$}
\label{Sec:6}
Let us first define the class of configurations called \textit{Formable Configurations $(FC)$}. 
\begin{definition}[Formable Configuration($FC$)]
    \label{def:FC}
    We say that a configuration is an $FC$ if the configuration is one of the following:
    \begin{enumerate}
        \item A single nominee configuration
        \item A double nominee configuration and $S_{Reflect} \ne \phi$
    \end{enumerate}
\end{definition}
 
 In a single nominee $FC$ we will follow the following notation.
 
 In a single nominee configuration if $r_0$ is the leader then the direction in which $r_0$ has the smallest angle sequence is called a \textit{Pivotal direction} and is denoted as $\mathcal{D}_p$. Let $r_i$ be the $i$-th robot from $r_0$ in the direction $\mathcal{D}_p$. Position of any robot $r_i$ is denoted as $R_i$ ( In some cases position of robots $r$ is denoted as $R$). Let $(R_i, R_{i+1})_{\mathcal{D}_p} = \alpha_i$ and $\sum_{i=0}^{n-1} \alpha_i = 2\pi$.

From Proposition~\ref{Prop:rotationallySymmetricImpossible} and Proposition~\ref{prop:reflectionalSymmetryImpossible} it is clear that if a distributed deterministic algorithm for APF has to be designed then only the $FC$s has to be considered as the initial configuration. So here in this section, assuming the initial configuration $\mathcal{C}(0)$ to be an $FC$, we propose an deterministic and distributed algorithm $APF\_CIRCLE$ that solves the arbitrary pattern formation problem on a continuous circle under Semi Synchronous (\textsc{SSync}) without chirality agreement by oblivious and silent robots with rigid move.
Algorithm $APF\_CIRCLE$ consists of several stages. In each  subsection of this section we describe the stages individually and provide the correctness proofs of the stage.   
 Now before moving further into the details of the algorithm let us define some special type of configurations that will be needed later.

\begin{definition}[Rotational Asymmetry Fixing Configuration ($RAFC$)]
    \label{def:RAFC}
     A configuration $\mathcal{C}$ is called a $RAFC$ if all the following conditions holds
     \begin{enumerate}
         \item $\mathcal{C}$ is a single nominee configuration with leader $r_0$.
         \item $\alpha_0 < \underset{i \ne 0}{\min}\{\alpha_i, \beta_0\}$ where $\beta_0$ is minimum angle in the target pattern.
     \end{enumerate}
\end{definition}

 \begin{definition}[RAFC Maintaining Configuration ($RMC$)]
     \label{def:RMC} A $RMC$ is a $RAFC$ where $\alpha_1 < \underset{j\ne 0,1}{\min}\{\alpha_j, \beta_0\}$ where $\beta_0$ is the minimum angle in target pattern.
 \end{definition}

 \begin{definition}[Partially Formed Configuration ($PFC$)]
\label{def:PFC}A single nominee configuration, say $\mathcal{C}$, is called a $PFC$ if all robots $r_p$, $p \ge 3$ are on their target location $T_p$ in $\mathcal{C}$.    
\end{definition}

\begin{figure}
\centering

\begin{subfigure}{.45\textwidth}
  \centering
  \includegraphics[height=5.5cm]{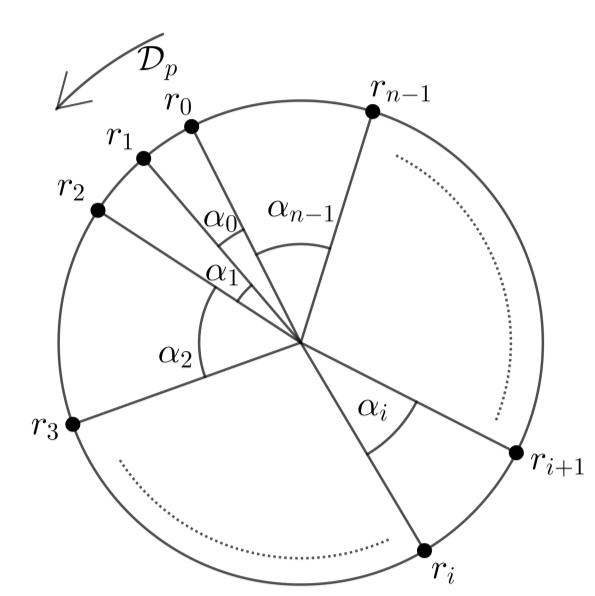}
  \caption{RAFC Maintaining Configuration. Here $\alpha_0< \underset{i\ne0}{\min}\{\alpha_i, \beta_0\}$ and $\alpha_1<~ \underset{j\ne0,1}{\min}\{\alpha_j, \beta_0\}$.}
  \label{fig:rmc}
\end{subfigure}%
\hfill
\begin{subfigure}{.45\textwidth}
  \centering
  \includegraphics[height=5.5cm]{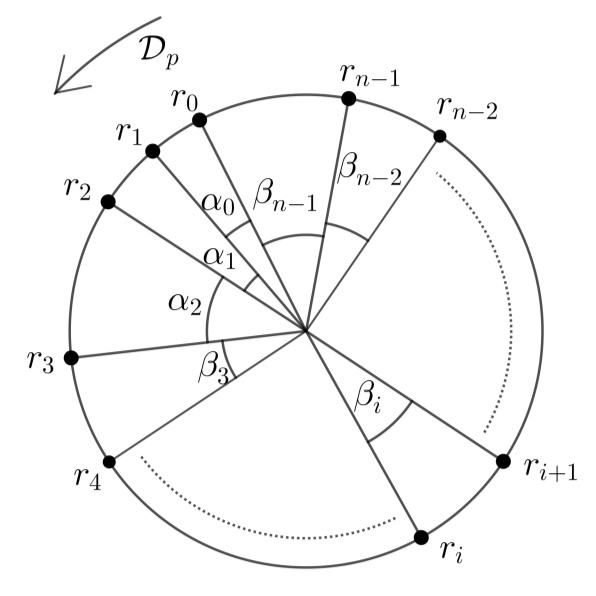}
  \caption{Partially Formed configuration, where all the robots $r_i$, $i\ge3$ are in target.}
  \label{fig:pfc}
\end{subfigure}
\caption{RMC and PFC}
\label{fig:rmcpfc}
\end{figure}

\subsection{A brief outline of the algorithm}
 The algorithm can be divided into seven stages based on the configuration. Since the robots are oblivious, in each Look-Compute-Move cycle, robots determine in which stage it belongs to by checking some certain conditions. This conditions are described by some Boolean variables in the following table:

\begin{center}
    \begin{tabular}{ |p{1.5cm}|p{5cm}|}
 \hline
 \textbf{Variable} & \textbf{Definition} \\
 \hline
 $c_0$  & Target pattern is formed\\
 \hline
 $c_1$ & Double nominee configuration\\
 \hline
 $c_2$ & $RAFC$\\
 \hline
 $c_3$ & $PFC$\\
 \hline
 $c_4$ & $RMC$\\
 \hline
 $c_5$ & $(R_1,R_2)_{\mathcal{D}_p} > \beta_0-\alpha_0$\\
 \hline
 $c_6$  & all but one robot is in target\\
 \hline
 $c_7$ & $(T_1, R_2)_{\mathcal{D}_p} \le \beta_1$\\
\hline
\end{tabular}
\end{center}

\par The main focus of the problem is to fix the position and the direction of the leader and to keep the leader and its pivotal direction fixed throughout the execution of the algorithm. Another motto is to move the robots to its target without collision. In Stage 1, leader moves an angular distance and makes the unique minimum angle. This minimum angle will be maintained throughout the algorithm so that the rotational asymmetry is maintained. 
After this the configuration becomes Rotational Asymmetry Fixing Configuration ($RAFC$). In Stage 2, the second neighbour $r_2$ of the leader in pivotal direction, moves an angular distance in direction $\mathcal{D}_p'$ and makes second minimum angle. It is ensured that the second minimum angle will not appear in the configuration throughout the execution of the algorithm. Thus Stage 2 ensures that the leader and its pivotal direction remain unchanged throughout the execution of the algorithm. These two stages are necessary, because if the rotational asymmetry is not maintained, then there might be a time when a configuration becomes a rotationally symmetric configuration and since the robots are oblivious, even if the initial configuration is Formable Configuration, the problem becomes unsolvable. Also if the leader or the pivotal direction changes infinitely often the robots may end up in a live-lock situation without forming target pattern. Note that after stage 2, configuration becomes a $RMC$ configuration.
\par
In stage 3, we ensure that all the robots except $r_0$, $r_1$, $r_2$ moves to their target avoiding collision and form a $PFC$. Since the position of the robots $r_0$ is assumed as the target position $T_0$, so after the $PFC$ formed only $r_1$ and $r_2$ are not in their target. From a $PFC$, for $r_1$ and $r_2$ to move to their target we have stage 4,5,6 and 7.
\par
If $r_2$ is in between $T_2$ and $T_3$, then $r_2$ moves to $T_2$ directly by executing stage 5. Note that, after execute stage 5, only $r_1$ is not in its target. In a $PFC$, if $r_2$ is in between $R_1$ and $T_1$, then $r_1$ can't move to $T_1$ unless $r_2$ moves. Now $r_2$ moves to target $T_2$ if after moving to $T_2$, angle between $R_1$ and $T_2$ is smaller than $\beta_{n-1}$. Otherwise the leader and the pivotal direction may change. For this case when angle between $R_1$ and $T_2$ is greater equal to $\beta_{n-1}$, then $r_2$ does not move to $T_2$ directly, it first moves to a point $D_2$ between $T_1$ and $T_2$ such that the angle between $R_1$ and $D_2$ is less than $\beta_{n-1}$. This ensures after this moves, the leader and its pivotal direction remain unaltered. This is done in stage 4. Note that after this stage , $r_1$ can now move to $T_1$ without collision. This is stage 6. We ensure that, in stage 6 also, the leader and the pivotal direction remain unchanged. After stage 6, either the pattern has been formed or there is only one robot $r_2$ which is not in its target position. For this case or case after the execution of stage 5, only one robot is not in its target. Then the robot which is not in its target, executes stage 7 and moves to its target. After completion of stage 7, the target pattern has been formed.

\subsection{Stage 1}


In Stage 1 the leader $r_0$ performs the following subroutine.
\\
\textbf{RAFC Formation():} 

\textbf{Input:} $(\neg c_0 \land c_1) \lor (\neg c_0 \land \neg c_1 \land \neg c_2 \land \neg c_3) \lor (\neg c_0 \land \neg c_1 \land \neg c_2 \land c_3 \land \neg c_6)$

\textbf{Output:} $\neg c_0 \land \neg c_1 \land c_2$
\\
If the configuration $\mathcal{C}$ is a single nominee configuration with the leader $r_0$,  then $r_0$ finds a real number $\epsilon_d$ and moves an angular distance  of $\epsilon_d$ in $\mathcal{D}_p$ such that the configuration becomes an $RAFC$. If $\mathcal{C}$ is a double nominee $FC$, the leader  $r_0$ has same view in both the direction. In this case the leader $r_0$ chooses any one of the direction, say $\mathcal{D}$ and moves an angular distance $\epsilon_d$ in $\mathcal{D}$ such that the configuration becomes an $RAFC$.

After $RAFC$ is formed, it is maintained unless $PFC$ is formed. This is needed to avoid  rotational symmetry. 
\begin{figure}
    \centering
    \includegraphics[height=4.5cm]{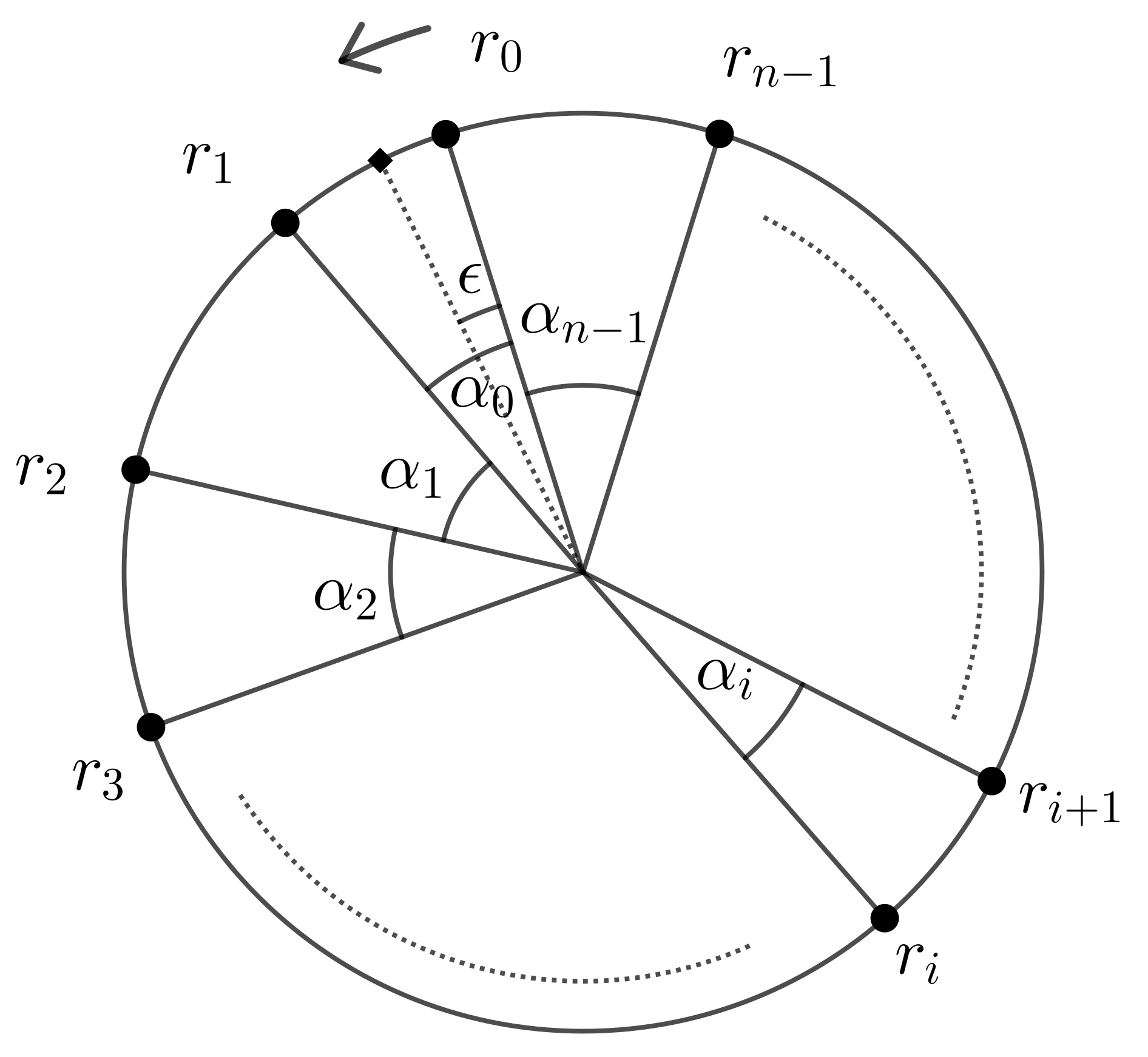}
    \caption{Here the leader $r_0$ moves towards $r_1$ and angular distance $\epsilon_d$ such that after the move the angular distance between $r_0$ and $r_1$ is strictly smallest in the current configuration and also strictly less than all angles in target pattern.}
    \label{fig:Leadermove}
\end{figure}

 The following lemma proves the existence of such $\epsilon_d$.
 
\subsubsection{Correctness of Stage 1}

 \begin{lemma}
     \label{lemma: Existence of epsilon d}
     If an $FC$ is not an $RAFC$, then there exists an $\epsilon_d >0$ such that after a move by the leader $r_0$ according to the subroutine \textit{RAFC Formation()} the configuration becomes an $RAFC$.
 \end{lemma}
 \begin{proof}
 Let the configuration $\mathcal{C}$ be an $FC$ which is not a $RAFC$. Then $\mathcal{C}$ is either a single nominee configuration or a double nominee configuration with $S_{Reflect} \ne \phi$.
 
     \textbf{Case: I}
     Let, the configuration $\mathcal{C}$ be a single nominee configuration with leader $r_0$. Since, $\mathcal{C}$ is not a $RAFC$, $\alpha_0 \ge \underset{i \ne 0}{\min}\{\alpha_i, \beta_0\}$.
     Note that, in $\mathcal{C}$, $\alpha_1 \le \alpha_{n-1}$ as $r_0$ is the nominee. Now for any $\epsilon_d \in (\alpha_0- \underset{i \ne 0}{\min}\{\alpha_i, \beta_0\}, \alpha_0) \subseteq (0, \alpha_0)$, if $r_0$ moves an angular distance $\epsilon_d$ in the direction $\mathcal{D}_p$, then $(R_0, R_1)_{\mathcal{D}_p}$ becomes $\alpha_0 - \epsilon_d$ which is strictly less than $\underset{i \ne 0}{\min}\{\alpha_i, \beta_0\}$ by the choice of the $\epsilon_d$. So now it is enough to show that after this move by $r_0$ the configuration, say $\mathcal{C}'$, remains a single nominee configuration. If possible let the configuration is not a single nominee configuration. Since $\mathcal{C}'$ is rotationally asymmetric, it must be a double nominee configuration where the nominees are the robots $r_0$ and $r_1$. Thus the minimum angle sequences of $r_0$ and $r_1$ in $\mathcal{C}'$ must be same. Now the second angles in the minimum angle sequences of $r_0$ and $r_1$ are  $\alpha_1$ and $\alpha_{n-1}+ \epsilon_d$ respectively. For them to be equal we have, $\epsilon_d = \alpha_1-\alpha_{n-1} \le 0$, which contradicts the choice of $\epsilon_d$. Hence $\mathcal{C'}$ must be a single nominee configuration and hence an $RAFC$. 

     \textbf{Case: II}
     Let the configuration $\mathcal{C}$ be a double nominee configuration with $S_{Reflect} \ne \phi$. Let $r_0 \in S_{Reflect} $ be  the leader in $\mathcal{C}$. Then $\mathcal{AS}_{D} (r_0) =\mathcal{AS}_{D'}(r_0)$. Since $\mathcal{C}$ is not an $RAFC$, $r_0$ moves in any one of clockwise or anticlockwise direction. Let the direction in which $r_0$ decided to move is denoted as $\mathcal{D}$. Let $r_i$ be the $i$-th robot from $r_0$ in the direction $\mathcal{D}$ and $(R_i, R_{i+1})_{\mathcal{D}} = \alpha_i$. Also in this configuration $\alpha_0 \ge \underset{i \ne 0}{\min}\{\alpha_i,\beta_0\}$. Let $\epsilon_d$ is chosen such a way that $\epsilon_d \in (\alpha_0-\underset{i \ne 0}{\min}\{\alpha_i,\beta_0\}, \alpha_0) \subseteq (0, \alpha_0)$ and $ \epsilon_d \ne \alpha_1-\alpha_{n-1}$. Note that existence of such $\epsilon_d$ is guaranteed from the fact that number of robots are finite and the interval from which $\epsilon_d$ is chosen has infinite points. Let $\mathcal{C'}$ be the configuration after $r_0$ moves. By the similar argument as in \textit{Case I}, it can be concluded that $\mathcal{C'}$ is an $RAFC$.

     Hence if a $FC$ is not an $RAFC$ then the leader $r_0$ can always find an $\epsilon_d >0$ such that after a move of angular distance $\epsilon_d$ by $r_0$ according to the subroutine \textit{$RAFC$ formation()} the configuration becomes an $RAFC$. 
 \end{proof}

\subsection{Stage 2}
 Observe that in a $RAFC$, the minimum angle $\alpha_0$ appears only once in the whole configuration. Also, in the algorithm we ensure that this angle is never formed anywhere else in the configuration. This implies that throughout the execution, the configuration always remains rotationally asymmetric until the pattern is formed. But this does not imply that the configuration will remain a $RAFC$ during the execution of the algorithm. This is needed because in an $RAFC$ we get an unique embedding of the target pattern. To do that, the second neighbour of $r_0$ in $\mathcal{D}_p$ in a $RAFC$ (i.e $r_2$) moves in the direction $\mathcal{D}_p'$ to form a \textit{RAFC maintaining configuration} or $RMC$.
 
If a $RAFC$ is not a $RMC$ then Stage 2 will be executed. In stage 2, $r_2$ performs the subroutine \textit{$RMC$ Formation()} described below.
\\
\textbf{RMC Formation():} 

\textbf{Input:} $\neg c_0 \land \neg c_1 \land c_2 \land \neg c_3 \land \neg c_4$

\textbf{Output:} $\neg c_0 \land \neg c_1 \land c_2 \land \neg c_3 \land c_4$
\\
If the configuration $\mathcal{C}$ is a $RAFC$ but not a $RMC$ and target is not formed then $r_2$ moves an angular distance $\epsilon_1$ in the direction $\mathcal{D}_p'$  such that the configuration  becomes an $RMC$.
\subsubsection{Correctness of Stage 2}
In the following lemma existence of such $\epsilon_1$ is guaranteed.
\begin{figure}[H]
     \centering
     \includegraphics[height=4.5cm]{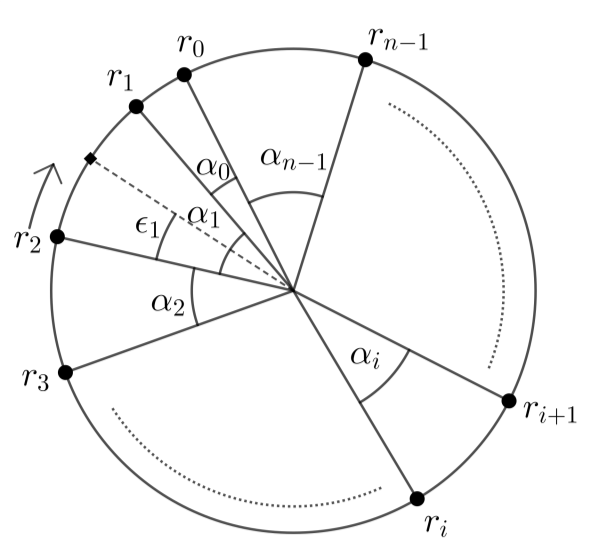}
     \caption{Here $\alpha_1$ is not strictly smaller than other $\alpha_i$s (except $i=0$) or $\beta_j$s. So $r_2$ moves an angular distance $\epsilon_1$ towards $r_1$ such that the new $\alpha_1$ after the move becomes the second uniquely minimum angle of the configuration and also less than all $\beta_j$s. Thus forming a $RMC$ }
     \label{fig:r2Moveanglemin}
 \end{figure} 
\begin{lemma}
    \label{lemma: Existence of epsilon 1}
    If a $RAFC$ is not a $RMC$ then there always exists an $\epsilon_1$ such that $\alpha_1-\alpha_0 > \epsilon_1 >0$ and after the move by $r_2$ of the angular distance $\epsilon_1$ in $\mathcal{D}_p'$ the configuration becomes an $RMC$.  
\end{lemma}
\begin{proof}
  Let $\mathcal{C}$ be a $RAFC$ which is not a $RMC$. Thus, $\mathcal{C}$ is a single nominee configuration with leader $r_0$ and $\alpha_0 < \underset{i \ne 0}{\min}\{\alpha_i, \beta_0\}$ and $\alpha_1 \ge \underset{j \ne 0,1}{\min}\{\alpha_j,\beta_0\}$ where $\beta_0$ is the minimum angle in the target pattern. Now to show the existence of an $\epsilon_1$ in the range $(0, \alpha_1-\alpha_0)$ such that if $r_2$ moves an angle $\epsilon_1$ in the direction $\mathcal{D}_p'$ the configuration becomes an $RMC$, it is enough to find an $\epsilon_1 \in (0,\alpha_1-\alpha_0)$ such that the condition
  $\alpha_1-\epsilon_1 < \underset{j\ne0,1}{\min}\{\alpha_j,\beta_0\}$ is true.

Let us consider $\epsilon_1 = \alpha_1 - \frac{\alpha_0 + \underset{i\ne 0}{\min}\{\alpha_i,\beta_0\}}{2}$.
This implies $\alpha_1-\epsilon_1 = \frac{\alpha_0 + \underset{i\ne 0}{\min}\{\alpha_i,\beta_0\}}{2} > \alpha_0$ (as $\alpha_0 < \underset{i \ne 0}{\min}\{\alpha_i, \beta_0\}$). Also, since $\alpha_0 < \alpha_1$ and $\underset{i \ne 0}{\min}\{\alpha_i, \beta_0\} \le \alpha_1$, $\epsilon_1 >0$. Now we only have to show that, the mentioned condition holds for our chosen $\epsilon_1$.  If possible let the condition does not hold for the chosen $\epsilon_1$. Now since $\underset{i\ne 0}{\min}\{\alpha_i,\beta_0\}\le \underset{j\ne 0,1}{\min}\{\alpha_j,\beta_0\}$,
    We have the following inequality,
    $$
    \frac{\alpha_0+\underset{i\ne0}{\min}\{\alpha_i,\beta_j\}}{2} = \alpha_1-\epsilon_1\ge\underset{j\ne0,1}{\min}\{\alpha_j,\beta_0\}\ge \underset{i\ne 0}{\min}\{\alpha_i,\beta_0\}
    $$
    Which implies, $\alpha_0\ge \underset{i\ne 0}{\min}\{\alpha_i,\beta_0\}$, a contradiction. Hence the condition $\alpha_1-\epsilon_1 < \underset{j\ne0,1}{\min}\{\alpha_j,\beta_0\}$ holds.   
\end{proof}

\subsection{Stage 3}
Note that in a $RMC$, the angles $\alpha_0$ and $\alpha_1$  occurs exactly once in the whole configuration. So, if these angles are not changed the configuration remains a single nominee configuration. Also, the leader and the pivotal direction remains same as we  ensure that these angles will not be formed again during the execution of this stage. So, after the $RMC$ is formed the target embedding remains unique. 

In this stage, the robots $r_i$, where $i \in \{3,4,\dots,n-1\}$ perform the subroutine\\ \textit{$PFC$ formation()} to eventually form a \textit{Partially formed Configuration} or $PFC$.

Before describing the subroutine 
\textit{$PFC$ formation()} we first need to define the term "\textit{Move Ready Robot}".

\begin{definition}[Move Ready Robot]
    In a $RMC$, let $r$ be the first robot from leader, say $r_0$ in the direction $\mathcal{D}_p$ which satisfies the following condition:
    \begin{enumerate}
        \item $r$ is not the first or second neighbour of leader $r_0$ in the direction $\mathcal{D}_p$.
        \item $(R,R')_{\mathcal{D}}-(R,T)_{\mathcal{D}}>\alpha_1$; where $T$ is the destination of $r$ in direction $\mathcal{D}$, $r'$ be the neighbour of $r$ in the direction $\mathcal{D}$, $R$ and $R'$ are the locations of $r$ and $r'$ respectively on the circle.
    \end{enumerate}
    Then $r$ is defined as the Move Ready robot of the configuration
\end{definition}
\textbf{PFC formation():} 

\textbf{Input:} $\neg c_0 \land \neg c_1 \land c_2 \land \neg c_3 \land c_4$

\textbf{Output:} $\neg c_0 \land \neg c_1 \land c_2 \land c_3$
\\
If a configuration $\mathcal{C}$ is not the target pattern and $\mathcal{C}$ is a $RMC$ which is not a $PFC$ then the move ready robot, say $r_p$ moves to the target position $T_p$. 
\\

\subsubsection{Correctness of Stage 3}
Observe that during this procedure no angle is created which is less or equal to $\alpha_0$ or $\alpha_1$. So during this procedure, the configuration remains a $RMC$ and the leader and the pivotal direction does not change.

We now have to ensure further that From a $RMC$ a $PFC$ will be formed eventually. To do so we have to prove that in a $RMC$ which is not a $PFC$, there will always be a robot which is the move ready robot. This will remove the possibility of a deadlock situation during the \textit{$PFC$ formation()} procedure. Thus in each round during the execution of this procedure, one of $r_p$ reaches its target position $T_p$ where $p \ge 3$. This implies \textit{$PFC$ formation()} runs for at most $n-3$ rounds and within this, the configuration will become a $PFC$. Now, to prove that a $RMC$ which is not a $PFC$ will have a robot which is the move ready robot, we have to first prove the following lemma.

\begin{lemma}
\label{lemma:Not Move Ready Same Direction}
    In a $RMC$ which is not a $PFC$, if a robot $r_i$, $i\ge 3$ is not a Move Ready robot and the destination of $r_i$ i.e., $T_i$ is in the direction $\mathcal{D}$ from $R_i$ then, the neighbor of $r_i$ in the direction $\mathcal{D}$, say $r_k$, must also have its target destination $T_k$ in direction $\mathcal{D}$ from $R_k$.
\end{lemma}
\begin{proof}
    Let $r_i$ ($i \ge 3$) be a robot that is not Move Ready and its destination $T_i$ is in direction $\mathcal{D}$ from $R_i$. Let $r_k$ be the neighbour of $r_i$ in the direction $\mathcal{D}$ ($k$ can be either $i+1$ or, $i-1$ in modulo n). Since $r_i$ is not move ready, $(R_i,R_k)_{\mathcal{D}}-(R_i, T_i)_{\mathcal{D}} \le \alpha_1$. Now there can be two possibilities. Either, $(R_i,R_k)_{\mathcal{D}}-(R_i, T_i)_{\mathcal{D}} \le 0$ or, $0 < (R_i,R_k)_{\mathcal{D}}-(R_i, T_i)_{\mathcal{D}} \le \alpha_1$. Now, if possible let the destination of $r_k$ i.e., $T_k$ be in the direction $\mathcal{D}'$ from $R_k$.

    \textbf{Case 1:} Let $(R_i,R_k)_{\mathcal{D}}-(R_i, T_i)_{\mathcal{D}} \le 0$. This implies $T_i$ is further than $R_k$ in the direction $\mathcal{D}$, from $R_i$ (Fig.\ref{fig:lemma5a}). Now consider $k = i+1$ and thus the $\mathcal{D}= \mathcal{D}_p$. Note that $i$ then can not be $n-1$ as $T_{n-1}$ can not be further than $R_0=T_0$ in the direction $\mathcal{D}_p$ from $R_{n-1}$ according to the target embedding. Now, for all other values for $i \ge 3$, if $T_{i+1}$ is in the direction $\mathcal{D}_p'$ from $R_{i+1}$, then $T_{i+1}$ appears before $T_i$ in the direction $\mathcal{D}_p$ in the embedding which is contradiction. Similarly, let us consider $k= i-1$ and thus $\mathcal{D} = \mathcal{D}_p'$. Here note that $i$ can not be $3$ as otherwise $r_3$ is Move Ready. This is because $T_3$ and $T_2$ must be on the arc from $R_2$  to $R_3$ in the direction $\mathcal{D}_p$. Thus $T_3$ can not be further than $R_2$ from $R_3$ in the direction $\mathcal{D}_p'$ as, $(T_3,R_2)_{\mathcal{D}_p'} > (T_3, T_2)_{\mathcal{D}_p'} =\beta_2 > \alpha_1 >0$ . Now for all other values of $i >3$, it can be shown that we will arrive at a contradiction by a similar argument as in the case where $k=i+1$ has been considered.

\begin{figure}
\centering
\begin{subfigure}{.5\textwidth}
  \centering
  \includegraphics[height=6.5cm]{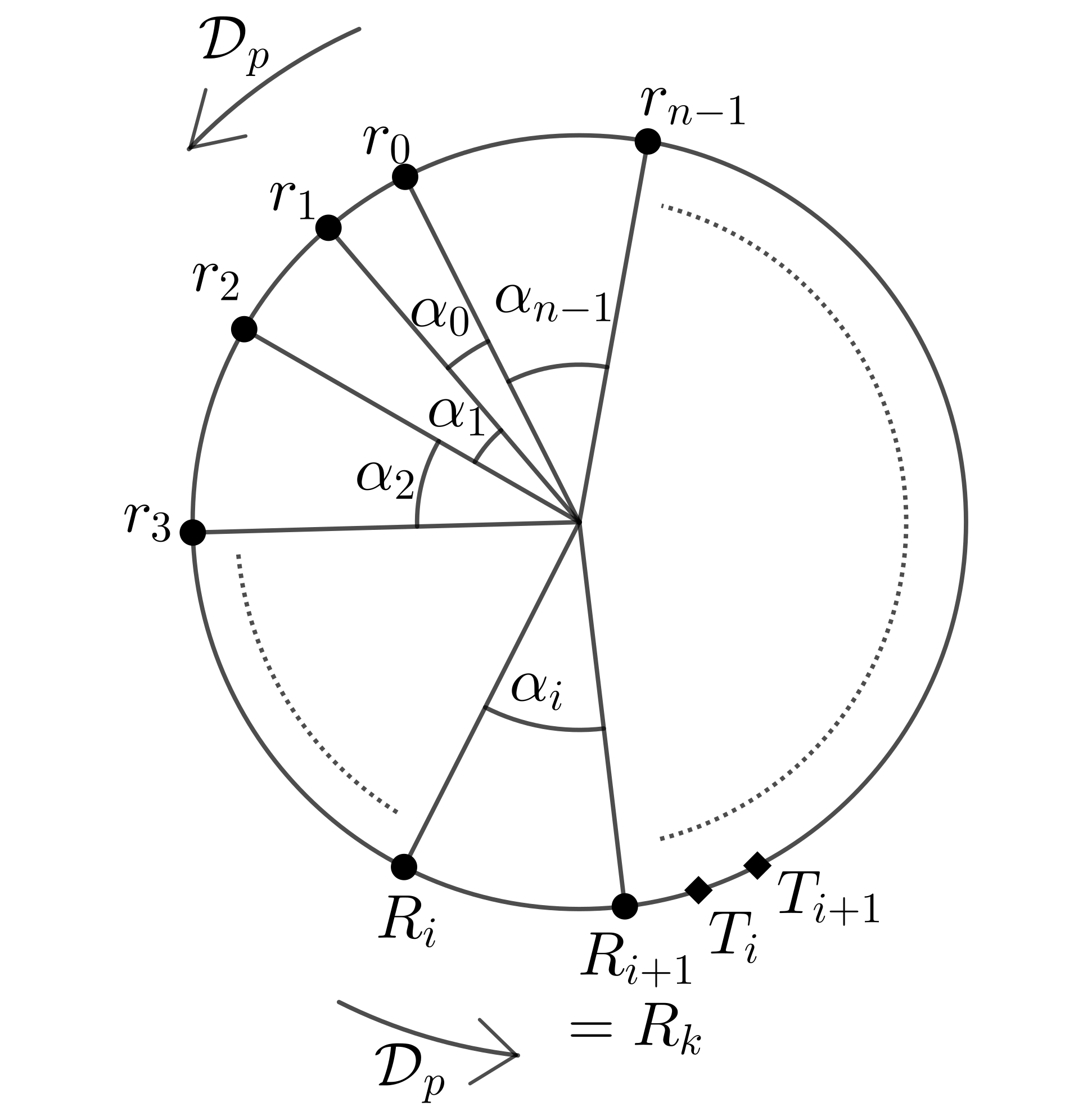}
  \caption{case:1}
  \label{fig:lemma5a}
\end{subfigure}%
\hfill
\begin{subfigure}{.5\textwidth}
  \centering
  \includegraphics[height=6.5cm]{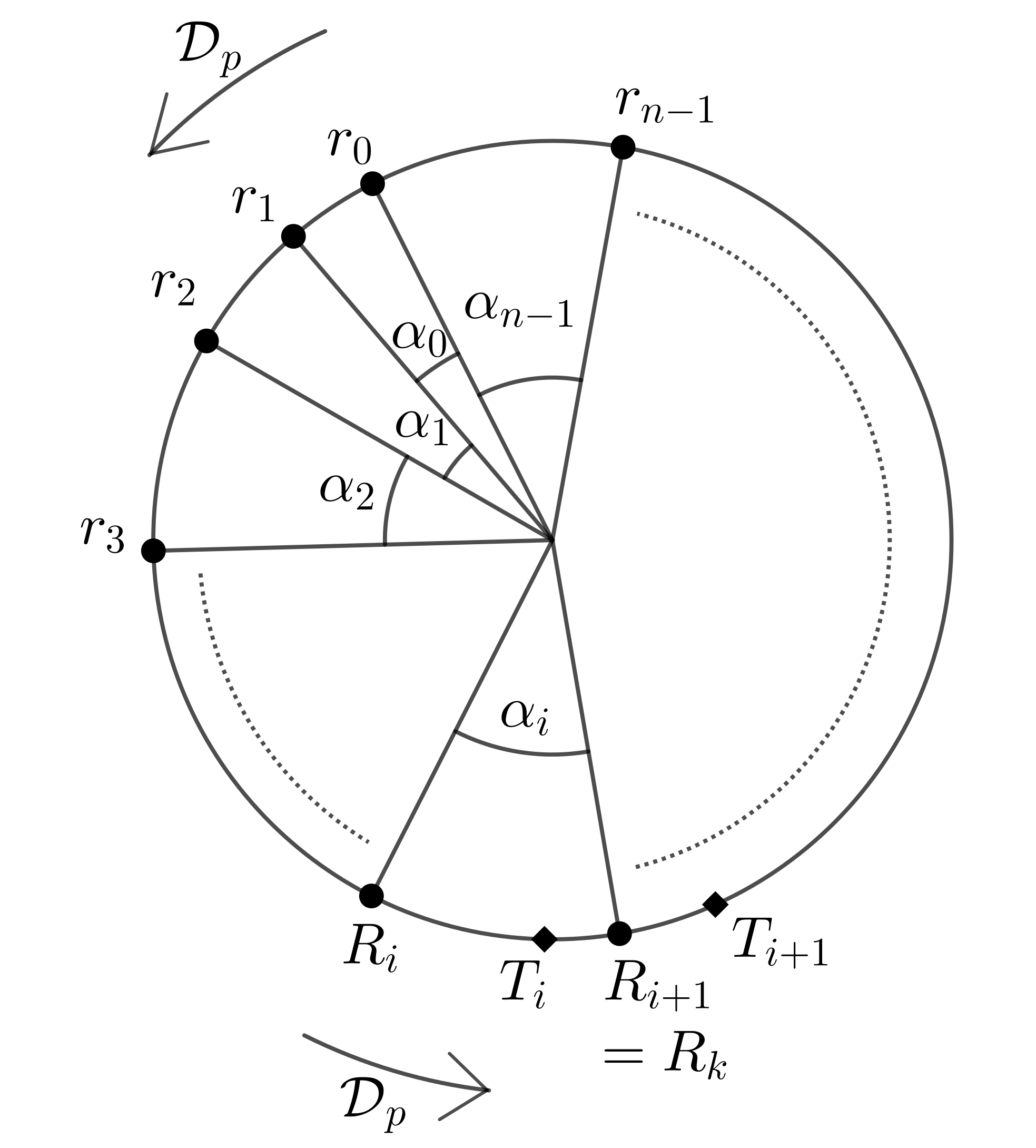}
  \caption{case:2}
  \label{fig:lemma5b}
\end{subfigure}
\caption{If the target destination $T_i$ of the robot $r_i$ is in the direction $\mathcal{D}$, then the target destination $T_{i+1}$ of its neighbour $r_{i+1}$ is also in the same direction $\mathcal{D}$. }
\label{fig:test}
\end{figure}

    \textbf{Case 2:} Let $0 < (R_i,R_k)_{\mathcal{D}}-(R_i, T_i)_{\mathcal{D}} \le \alpha_1$. This implies $R_k$ is further than $T_i$ from $R_i$ in the direction $\mathcal{D}$ but, $(T_i, R_k)_{\mathcal{D}} \le \alpha_1$ (Fig.~\ref{fig:lemma5b}). Let $k= i+1$ and hence $\mathcal{D}= \mathcal{D}_p$ ($i$ can not be $n-1$ as shown earlier in case 1). Now, according to the embedding $T_i$ can not be further than $T_{i+1}$ from $T_0$ in the direction $\mathcal{D}_p$. Hence, $T_{i+1}$ must be on the arc joining from $T_i$ to $R_{i+1}$ in the direction $\mathcal{D}_p$. This implies $\beta_0 \le (T_i, T_{i+1})_{\mathcal{D}_p} \le (T_i, R_{i+1})_{\mathcal{D}_p} \le \alpha_1 \implies \beta_0 \le \alpha_1$, a contradiction due to the fact that the configuration is a $RFC$. Similarly if $k=i-1$ and hence the direction $\mathcal{D}= \mathcal{D}_p'$ then again we will arrive at a contradiction by a similar argument. 

    Since for both the possibilities we arrive at a contradiction, $T_k$ must also be in the direction of $\mathcal{D}$ from $R_k$.
    
\end{proof}

\begin{lemma}
\label{lemma:ATleast one Move Ready}
    If a $RMC$ is not a $PFC$ then there exists a robot $r_p$ which is Move Ready.
\end{lemma}
\begin{proof}
A robot is called terminated if it has already reached its target. If possible let in a $RFC$ the robots $r_i$ ($i \ge 3$) are either terminated or not Move Ready (existence of such robot is guaranteed by the fact that the configuration is not a $PFC$). Let $r_k$ be a  robot from $R_0$ in the direction $\mathcal{D}_p$ which has not terminated and is not Move Ready. Let the target of $r_k$ i.e., $T_k$ be in a direction $\mathcal{D}$ from $R_k$. Observe that if $r_k = r_3$, then $\mathcal{D} = \mathcal{D}_p$. Otherwise, since $T_2$ is in the direction $\mathcal{D}_p$ from $R_2$, $(R_3,R_2)_{\mathcal{D}_p'}- (R_3, T_3)_{\mathcal{D}_p'} = (R_2,T_3)_{\mathcal{D}_p}\ge (T_2,T_3)_{\mathcal{D}_p} = \beta_2 > \alpha_1$ and hence $r_3$ becomes Move Ready. Similarly if $r_k = r_{n-1}$ then, $T_{n-1}$ must be in the direction $\mathcal{D}_p'$ from $R_{n-1}$. Otherwise, $T_{n-1}$ must lie on the arc joining $R_{n-1}$ and $T_0 = R_0$ in the direction $\mathcal{D}_p$ which implies $(R_{n-1},R_0)_{\mathcal{D}_p}-(R_{n-1},T_{n-1})_{\mathcal{D}_p} = (T_{n-1},T_0)_{\mathcal{D}_p} = \beta_{n-1} > \alpha_1$ a contradiction. 

Now we claim that, for a robot $r_i, (i \ge 3)$ which has not terminated and is not move ready, if the direction of its target is in the direction $\mathcal{D}$ from $r_i$, then the neighbor of $r_i$ , say $r_j$ in the direction $\mathcal{D}$ must have not terminated also. Otherwise, if $r_j$ is terminated then it must be on $T_j$. Also, $T_i$ must be on the arc joining the points from $R_i$ to $T_j$ in the direction $\mathcal{D}$. This implies $(R_i,R_j)_{\mathcal{D}}-(R_i,T_i)_{\mathcal{D}} = (T_i,T_j)_{\mathcal{D}} = \beta_t >\alpha_1$ ($t \in \{i,j\}$) and thus $r_i$ becomes move ready contrary to the assumption. 

So, now for a robot $r_{k_1}$ which is not Move Ready and has not terminated yet, let $\mathcal{D}$ be the direction of $T_{k_1}$ from $R_{k_1}$ ($R_{k_1}$ is the location of $r_{k_1}$ on the circle). Also let $r_{k_2}$ be the neighbour of $r_{k_1}$ in the direction $\mathcal{D}$. By Lemma~\ref{lemma:Not Move Ready Same Direction} and the above claim $r_{k_2}$ must have not terminated yet and the direction of $T_{k_2}$ must be in the direction $\mathcal{D}$ from $R_{k_2}$ ($R_{k_2}$ is the location of $r_{k_2}$ on the circle). Now by mathematical induction, it can be shown that all robots $r_i$ ($i \in \{3,4,\dots,n-1\}$) in the direction $\mathcal{D}$ from $r_k$, must have not terminated and are not Move Ready. So either $r_3$ or $r_{n-1}$ must be not Move Ready and has not terminated. If $r_3$ is not move ready and has not terminated then the direction of $T_3$ must be $\mathcal{D}_p$ from $R_3$ and then by induction it can be shown that $r_{n-1}$ must also be not Move ready and has not terminated and direction of $T_{n-1}$ must be in $\mathcal{D}_p$ from $R_{n-1}$ which is a contradiction. Similarly, if $r_{n-1}$ is not Move ready and has not terminated then $T_{n-1}$ must be in the direction $\mathcal{D}_p'$ from $R_{n-1}$ which will imply $r_3$ is not Move Ready and is not terminated and $T_3$ must be in direction $\mathcal{D}_p'$ from $R_3$. which is again a contradiction. Hence in a $RMC$ which is not a $PFC$ there always exists a robot which is the move ready robot in the configuration.    
\end{proof}

\subsection{Stage 4}
After the configuration becomes a $PFC$ after completion of Stage 3, only the robots $r_1$ and $r_2$ are not in their target locations. In this scenario the condition $\neg c_6$ true. Now, if $(R_1,R_2)_{\mathcal{D}_p} \le \beta_0-\alpha_0$, 
then the robot $r_2$ performs the procedure \textit{ R2Move1()}.
\\
\textbf{R2Move1():}

\textbf{Input:} $\neg c_0 \land \neg c_1 \land c_2 \land c_3 \land \neg c_6 \land \neg c_5$

\textbf{Output:} $(\neg c_0 \land \neg c_1 \land c_2 \land c_3 \land \neg c_6 \land c_5) \lor (\neg c_0 \land \neg c_1 \land c_2 \land c_3 \land c_6)$
\\
If $\beta_0 +\beta_1-\alpha_0 < \beta_{n-1}$ then $r_2$ moves to $T_2$ otherwise chose a $\delta \in (0, \beta_{n-1}-\beta_0+\alpha_0)$ and move to an angular distance $\alpha_0 +\beta_{n-1}- \delta$ from $R_0 = T_0$ in $\mathcal{D}_p$.

\subsubsection{Correctness of Stage 4}

\begin{lemma}
\label{lemma:Move of R2 dont Change Leader}
    If $r_2$ executes the procedure \textit{R2Move1()} then the configuration remains a single nominee configuration where the leader and pivotal direction does not change.
\end{lemma}
\begin{proof}
When $r_2$ executes \textit{R2Move1()} during stage 4, $r_1$ is not on $T_1$, and $r_2$ moves to $D_2$. Now, either $D_2=T_2$ or, it is a point on the circle such that,  $(R_1,D_2)_{\mathcal{D}_p}= (R_0,D_2)_{\mathcal{D}_p}-(R_0,R_1)_{\mathcal{D}_p}= \alpha_0+\beta_{n-1}-\delta -\alpha_0 =\beta_{n-1}-\delta$, where $0 <\delta<\beta_{n-1}-\beta_0+\alpha_0$. As $\alpha_0<\alpha_1<\beta_0\le\beta_1$, then the destination $D_2$ of $r_2$ must be on the arc joining from point $R_2$ to point $R_3=T_3$ in the direction $\mathcal{D}_p$. 

\textit{Case-I:} If $(R_1,T_2)_{\mathcal{D}_p}=\beta_0+\beta_1-\alpha_0<\beta_{n-1}$, then $r_2$ moves to its target $T_2$. Then the angle sequence of $r_0$ in the pivotal direction remains uniquely minimum, as $(R_1,T_2)_{\mathcal{D}_p}<\beta_{n-1}$. Thus in this case the configuration remains a single nominee configuration and also the leader and the pivotal direction does not change.

\textit{Case-II:} If $(R_1,T_2)_{\mathcal{D}_p}=\beta_0+\beta_1-\alpha_0 \ge \beta_{n-1}$, then $r_2$ moves to a point $D_2$ in the direction $\mathcal{D}_p$ such that $(R_1,D_2)_{\mathcal{D}_p}$ must be $\beta_{n-1}-\delta$. Then the minimum angle sequence of the configuration remains unique and belongs to $\mathcal{AS}(r_0)$ and the pivotal direction also remains same as $\beta_{n-1}>\beta_{n-1}-\delta$, for any $\delta \in (0, \beta_{n-1}-\beta_0+\alpha_0)$. So the configuration remains a single nominee configuration and also the leader and the pivotal direction does not change. 
\end{proof}

\begin{lemma}
\label{lemma:r1 Moves}
    If $r_2$ executes \textit{R2Move1()} in stage 4 and $R_2$ is the new position of $r_2$ after the move then the condition $(R_1,R_2)_{\mathcal{D}_p} > \beta_0-\alpha_0$ must become true.
\end{lemma}
\begin{proof}
    Let $r_2$ executes \textit{R2Move1()} in stage 4 and let $R_2$ be the new position of $r_2$ after the move. Now $R_2$ is either $T_2$ or, the point  $D_2$ such that $(R_1,D_2)_{\mathcal{D}_p} = \beta_{n-1}-\delta$ for some $\delta \in (0, \beta_{n-1}-\beta_0+\alpha_0)$.
    
    \textit{Case-I:} Let $R_2$ is $T_2$. We now have to show that $(R_1, T_2)_{\mathcal{D}_p} > \beta_0-\alpha_0$.
    If $(R_1, T_2)_{\mathcal{D}_p} \le \beta_0-\alpha_0$ then, $\beta_0+\beta_1-\alpha_0 \le \beta_0-\alpha_0 \implies  \beta_1\le 0 $ which is a contradiction. Hence  $(R_1, T_2)_{\mathcal{D}_p} > \beta_1$.

    \textit{Case-II:} Let $R_2$ is $D_2$ such that $(R_1, D_2)_{\mathcal{D}_p} = \beta_{n-1}-\delta$, where $\delta \in (0, \beta_{n-1}-\beta_0+\alpha_0)$. This implies $(R_1,D_2)_{\mathcal{D}_p} = \beta_{n-1}-\delta > \beta_0-\alpha_0$ and hence the result.
    
\end{proof}

\subsection{Stage 5}

Stage 5 is executed if target is not already formed and the current configuration is a $RAFC$ and a $PFC$ with $(R_1, R_2)_{\mathcal{D}_p} > \beta_0-\alpha_0$ and $(T_1, R_2) > \beta_1$ i.e., $r_2$ is located in the arc joining $T_2$ and $T_3$ in the direction $\mathcal{D}_p$.
In this stage the robot $r_2$ executes the procedure \textit{R2MoveReverseToT2()}.
We describe the procedure in the following.
\\
\textbf{R2MoveReverseToT2():}

\textbf{Input:} $\neg c_0 \land \neg c_1 \land c_2 \land c_3 \land \neg c_6 \land  c_5 \land \neg c_7$

\textbf{Output:} $\neg c_0 \land \neg c_1 \land c_2 \land c_3 \land c_6$
\\
In this stage the robot $r_2$ moves to $T_2$ in the direction $\mathcal{D}_p'$.

\subsubsection{Correctness of Stage 5}

\begin{lemma}
    \label{lemma:stage5} If $r_2$ executes the procedure \textit{R2MoveReverseToT2()} in stage 5, the configuration remains a single nominee configuration where leader and the pivotal direction does not change.
\end{lemma}
\begin{proof}
    Let the configuration becomes a double nominee configuration after $r_2$ executes \textit{R2MoveReverseToT2()} in stage 5 during some round, say $t$. Since the configuration remains an $RAFC$ after completion of the round only $r_0$ and $r_1$ can be the nominees. Let $\mathcal{D}_p$ be the pivotal direction at the beginning of the round $t$. Then $\mathcal{AS}_{\mathcal{D}_p}(r_0)$ decreases more after the completion of round $t$. Now if the new $\mathcal{AS}_{\mathcal{D}_p}(r_0) = \mathcal{AS}_{\mathcal{D}_p'}(r_1)$ then, at the beginning of round $t$, $\mathcal{AS}_{\mathcal{D}_p}(r_0) > \mathcal{AS}_{\mathcal{D}_p'}(r_1)$. This is a contradiction. Hence, $r_1$ can not be a nominee and thus the configuration remains a single nominee configuration and the leader as well as the pivotal direction remains same.
\end{proof}

\subsection{Stage 6}

This stage is executed if the target pattern is not already formed and the current configuration is a $RAFC$ and a $PFC$ along with $(R_1, R_2)_{\mathcal{D}_p} > \beta_0-\alpha_0$ and $(T_1,R_2)_{\mathcal{D}_p} \le \beta_1$. In this stage the robot $r_1$ executes the procedure \textit{R1MoveToTarget()}. We describe the procedure in the following.
\\
\textbf{R1MoveToTarget():}

\textbf{Input:} $\neg c_0 \land \neg c_1 \land c_2 \land c_3 \land \neg c_6 \land  c_5 \land  c_7$

\textbf{Output:} $\neg c_0 \land \neg c_1  \land c_3 \land c_6$
\\
In this procedure $r_1$ moves to $T_1$.

\subsubsection{Correctness of Stage 6}

\begin{lemma}
    \label{lemma: r1 no collision}
    By following the procedure \textit{R1MoveToTarget()}, $r_1$ can move to $T_1$ without collision.
\end{lemma}
\begin{proof}
  Collision occurs only when $(R_1,T_1)_{\mathcal{D}_p} \ge (R_1,R_2)_{\mathcal{D}_p} > \beta_0-\alpha_0$.
  This implies, $\beta_0-\alpha_0 > \beta_0-\alpha_0$. Thus we reach a contradiction. Hence $r_1$ moves to $T_1$ by executing \textit{R1MoveToTarget()} in stage 6 without collision.
\end{proof}
After completion of stage 6, two possible things can happen. Either the target is formed or, all robot but $r_2$ are in their corresponding  target position . Now for the latter case, we have to ensure that the configuration after $r_1$ executes \textit{R1MovesToTarget()}, remains a single nominee configuration. Otherwise, an unique embedding can not be agreed upon by the robots. Since $(T_1,R_2)_{\mathcal{D}_p} < \beta_1$ (for the latter case) before the execution of stage 6, $r_2$ must be located on the arc joining $T_1$ and $T_2$ in the direction $\mathcal{D}_p$ but not on $T_2$. Let after $r_1$ moves to $T_1$ during Stage 6, $(T_1,R_2)_{\mathcal{D}_p} = \beta_1-\epsilon$, and $(R_2,T_3)_{\mathcal{D}_p} = \beta_2+\epsilon$ for some $\epsilon > 0$ and $\mathcal{D}_p$ being the pivotal direction before execution of stage 6 by $r_1$.   Now since in this stage the configuration deviates from being a $RAFC$, we first have to ensure that after completion of this stage rotational symmetry does not occur. The following lemma ensures it.
\begin{lemma}
    If $r_1$ performs \textit{R1MoveToTarget()} in stage 6, the configuration does not become rotationally symmetric.
\end{lemma}
\begin{proof}
    If possible let the configuration become rotationally symmetric after the movement of $r_1$. Then there is another robot which has an angle sequence same as $r_0$.Let $r_i$ be that robot.

    \textit{Case I:} $i \ge 3$. In this case $r_i$ will have strictly smallest angle sequence in the target pattern, which is not true according to our embedding.
    
    \textit{Case II:} $i=1$. For this case, all angles in the configuration is $\beta_0$. This gives $(R_2,R_3)_{\mathcal{D}_p}=(R_2,T_3)_{\mathcal{D}_p} = \beta_0$. Thus, $\beta_2+\epsilon=\beta_0$, implies $\beta_2<\beta_0$. This is a contradiction.

    \textit{Case III:} $i=2$. For this case, similar to the previous case, $(R_2,R_3)_{\mathcal{D}_p}=\beta_0$, which similarly leads to contradiction. 
    
    Therefore, the configuration remains rotationally asymmetric. 
\end{proof}

 Now, for a unique embedding we also have to ensure that the new configuration is a single nominee configuration. For that we have the following lemma.
\begin{lemma}
     If $r_1$ performs \textit{R1MoveToTarget()} in stage 6, the configuration remains a single nominee configuration. Also the leader and the pivotal direction does not change.
\end{lemma}
\begin{proof}
    If the configuration after $r_1$ executes stage 6 becomes a double nominee configuration then $r_0$ and $r_1$ can only be the nominees. This is because no other robot has $\beta_0$ and $\beta_1-\epsilon$ as the first two terms of its angle sequence. Also note that, first two terms of $\mathcal{AS}_{\mathcal{D}_p}(r_0)$ and $\mathcal{AS}_{\mathcal{D}_p}(r_1)$ respectively are $(\beta_0, \beta_1-\epsilon)$ and $(\beta_0,\beta_{n-1})$ ($\mathcal{D}_p$ is the pivotal direction before execution of stage 6). Now since $\beta_{n-1} \ge \beta_1 > \beta_1-\epsilon$, $\mathcal{AS}_{\mathcal{D}_p}(r_0) < \mathcal{AS}_{\mathcal{D}_p}(r_1)$.  Thus we arrive at a contradiction. Hence after $r_1$ executes \textit{R1MoveToTarget()} in stage 6 the configuration remains a single nominee configuration with same leader and the pivotal direction.
\end{proof}

\subsection{Stage 7}
This stage executes only when the configuration is a single nominee configuration and all but one robot are at their corresponding target locations.
During this stage the robot that is not at the target moves to target by executing the procedure \textit{Target Formation()}.
\\
\textbf{Target Formation():}

\textbf{Input:} $\neg c_0 \land \neg c_1 \land c_3 \land c_6$

\textbf{Output:} $c_0$
\\
the robot which is not in target moves to its target location thus forming the pattern.

Note that in the worst case all stages but Stage 3 takes only one epoch too terminate. And stage 3 takes at most $n-3$ epochs to terminate . Also observe in the algorithm flowchart (Fig.~\ref{fig:flowchart1}, Fig.~\ref{fig:flowchart2}) that no stage is executed more than once. Also in the whole execution of the algorithm if stage 5 is executed then stage 6 will not be executed and vice versa. So the total time taken by algorithm $APF\_CIRCLE$ to terminate is $n+2$ epochs. 
Thus we have the following theorem stating the correctness of the algorithm $APF\_CIRCLE$.
\begin{theorem}
    \label{theorem:correctness}
    Algorithm $APF\_CIRCLE$ can solve arbitrary pattern formation problem on a continuous circle with oblivious and silent swarm of robots without chirality and under a semi synchronous scheduler from any Formable Configuration (FC) within $O(n)$ epochs where $n$ is the number of robots in the swarm.
\end{theorem}
\begin{figure}[H]

    \centering
    \includegraphics[height=7cm,width=15cm]{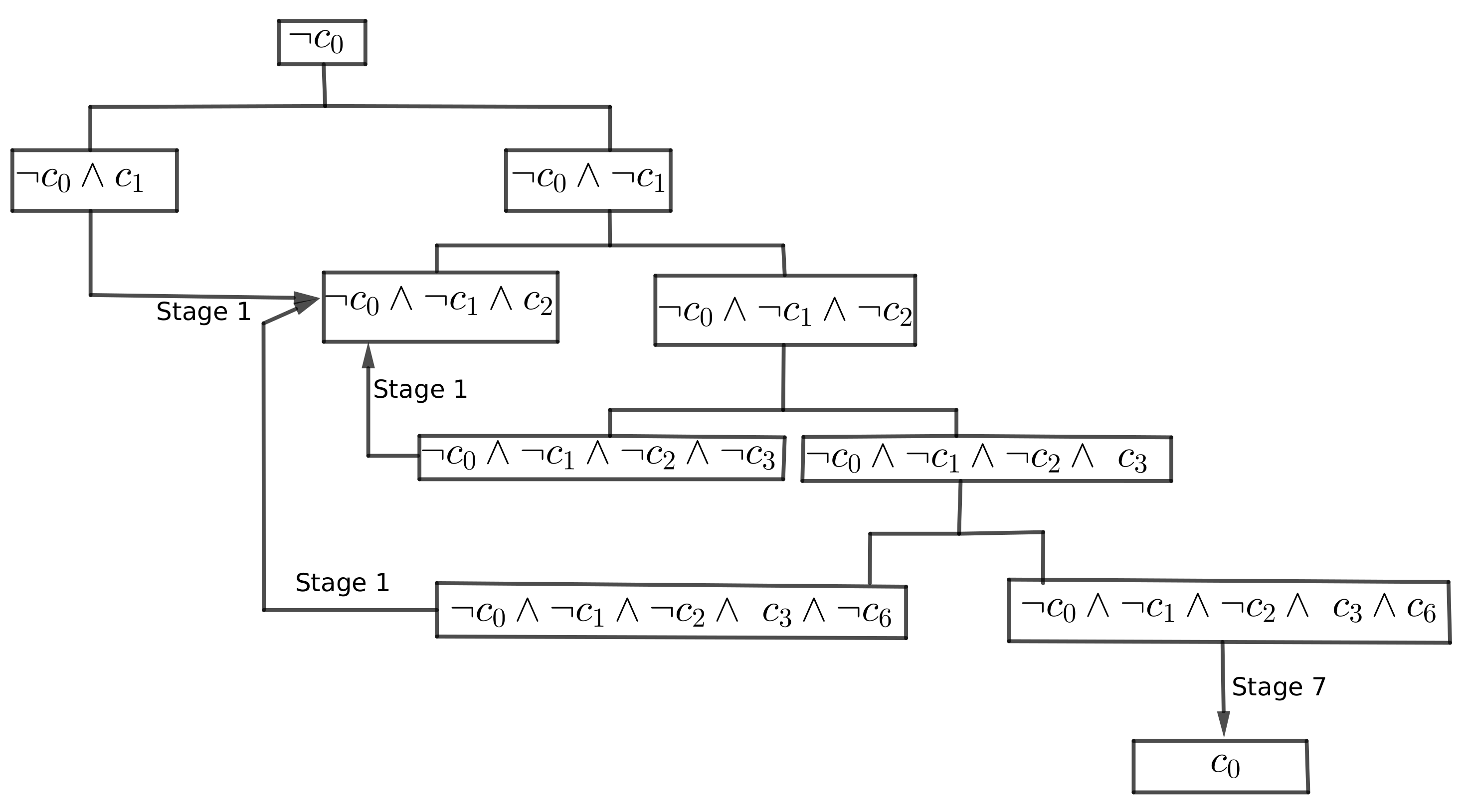}
    \caption{Algorithm Flowchart: part 1}
    \label{fig:flowchart1}
\end{figure}
\begin{figure}[H]

    \centering
    \includegraphics[height=10cm, width=15cm]{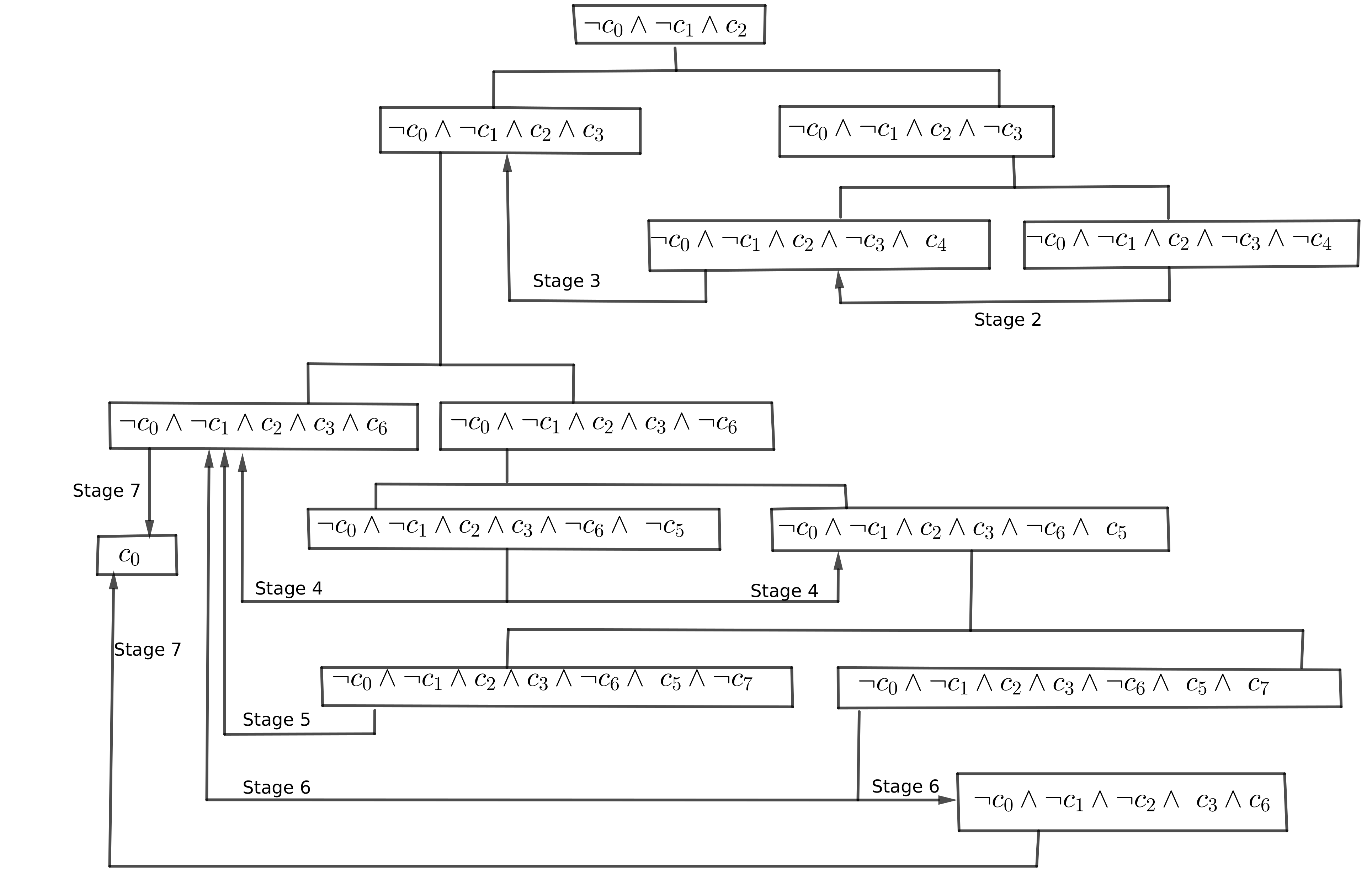}
    \caption{Algorithm Flowchart: part2}
    \label{fig:flowchart2}
\end{figure}
\section{Conclusion}
\label{Sec:7}
The arbitrary pattern formation problem is a classical problem in the field of swarm robotics. Till now it has been studied considering the euclidean plane and some discrete domains mostly. In continuous domains, there are certain environments that restrict the movement of the robot in any direction. Any closed curve embedded on a plane is an example of this. In the real world, this kind of environment can be seen everywhere, for example, train lines, road networks, etc. It can be argued that a problem solvable in a continuous circle can be solved on any closed curve. So here, in this paper, this problem has been introduced on a continuous circle for the first time.  Her in this work we have completely characterized the class of initial configurations for which arbitrary pattern formation problem is solvable in a deterministic method and then provided an deterministic and distributed algorithm $APF\_CIRCLE$ which solves the APF problem for any solvable configuration considering the robots to be oblivious, silent and without chirality under a semi synchronous scheduler.

For the days ahead, it would be really interesting if this problem can be solved under an asynchronous scheduler. Also, another interesting thing  would be to find out if there is an initial configuration and a target configuration such that for any embedding of the target the time taken by the $n$ robots to form the target is $O(n)$ or not. If this lower bound is $O(n)$ then the algorithm presented here is time optimal otherwise another time-optimal algorithm has to be designed. One can also consider the limited visibility model to study this problem to extend this research further.  

\bibliography{apfcir}

\newpage
\end{document}